\documentclass{article}

\usepackage{proof}
\usepackage{url}
\usepackage{amssymb}
\usepackage{latexsym}
\usepackage{graphics}

\newenvironment{proof}{{\it Proof}.\begin{trivlist}\item}
{\hfill{$\Box$}\end{trivlist}}

\title{Small-step and big-step semantics for call-by-need}
\author
{Keiko Nakata\thanks{Supported by the Estonian Science Foundation 
grant no.~6940 and the ERDF cofunded project EXCS, the Estonian Centre 
of Excellence in Computer Science.}
\\ 
Institute of Cybernetics, Tallinn University of Technology\\
Masahito Hasegawa\thanks{Partly supported by the Grant-in-Aid 
for Scientific Research (C) 20500010.}\\
Research Institute for Mathematical Sciences, Kyoto University}
\date{}

\newcommand{\needmapsto}{%
\mbox{$\,\displaystyle\mathop{\longrightarrow}_{\mbox{\tiny\sf NEED}}\,$}}

\newcommand{\valuemapsto}{%
\mbox{$\,\displaystyle\mathop{\longrightarrow}_{\mbox{\tiny\sf VALUE}}\,$}}

\newcommand{\acyclic}{\ensuremath{\lambda_{\mathit{let}}}}
\newcommand{\cyclic}{\ensuremath{\lambda_{\mathit{letrec}}}}
\newcommand{\comment}[1]{}
\newcommand{\remarque}[1]{}

\newcommand{\fun}[2]{\ensuremath{\lambda #1.#2}}
\newcommand{\app}[2]{\ensuremath{#1 #2}}
\newcommand{\letin}[2]{\ensuremath{\mathsf{let}~#1~\mathsf{in}~#2}}
\newcommand{\evalin}[2]{\ensuremath{\langle #1\rangle\; #2}}
\newcommand{\evalfresh}[5]{\ensuremath{\evalin{#1}{#2}\Downarrow_{#5}\evalin{#3}{#4}}}
\newcommand{\eval}[4]{\ensuremath{\evalin{#1}{#2}\Downarrow\evalin{#3}{#4}}}

\newcommand{\Eval}[4]{\ensuremath{\vdash \evalin{#1}{#2}\Downarrow\evalin{#3}{#4}}}
\newcommand{\maps}[2]{\ensuremath{#1 \mapsto #2}}
\newcommand{\rt}{\ensuremath{\epsilon}}
\newcommand{\subst}[3]{\ensuremath{#1[#2/#3]}}

\newcommand{\hpleq}[2]{\ensuremath{#1 \leq #2}}

\newcommand{\LBVs}[1]{\ensuremath{\mathit{LBV}(#1)}}
\newcommand{\FVs}[1]{\ensuremath{\mathit{FV}(#1)}}
\newcommand{\Body}[1]{\ensuremath{\mathit{Exp}(#1)}}
\newcommand{\prd}[2]{\ensuremath{(#1, #2)}}

\newcommand{\rar}{\ensuremath{\rightarrow}}
\newcommand{\rarr}{\ensuremath{\twoheadrightarrow}}
\newcommand{\rarR}{\ensuremath{\mapsto\!\!\!\!\!\!\!\rightarrow}}

\newcommand{\comp}[1]{\ensuremath{\lfloor#1\rfloor}}
\newcommand{\decomp}[1]{\ensuremath{\lceil#1\rceil}}
\newcommand{\fresh}{{\rm fresh}}
\newcommand{\dom}[1]{\ensuremath{\mathit{dom}(#1)}}
\newcommand{\ran}[1]{\ensuremath{\mathit{ran}(#1)}}

\newcommand{\error}{\ensuremath{\bullet}}
\newcommand{\letrec}[2]{\ensuremath{{\sf let~rec} ~#1 ~{\sf in} ~#2}}
\newcommand{\be}[2]{\ensuremath{#1~\mathsf{be}~#2}}
\newcommand{\dps}[2]{\ensuremath{D[#1,#2]}}

\newcommand{\update}[3]{\ensuremath{#1[#2\mapsto #3]}}
\newcommand{\updateseq}[4]{\ensuremath{#1[#2\mapsto #3]_{#4}}}

\newcommand{\flatten}[1]{\ensuremath{\overline{#1}}}

\newcommand{\fleq}[2]{\ensuremath{#1 \leq_{\mathcal{F}} #2}}

\newcommand{\uni}[2]{\ensuremath{#1 \cup #2}}

\newcommand{\cyclicval}{\ensuremath{\lambda^{\mathit{val}}_{\mathit{letrec}}}}

\newcommand{\prj}[2]{\ensuremath{\pi_{#2}(#1)}}

\newtheorem{lemma}{Lemma}[section]
\newtheorem{proposition}{Proposition}[section]
\newtheorem{theorem}{Theorem}[section]
\newtheorem{corollary}{Corollary}[section]

\newcommand{\emp}{\ensuremath{\emptyset}}

\newcommand{\restr}[2]{\ensuremath{#1|_{\overline{#2}}}}
\newcommand{\lift}{\ensuremath{\mathit{Fn}}}
\newcommand{\proj}{\ensuremath{\downarrow_{\mathit{Fn}}}}
\newcommand{\Values}{\ensuremath{\mathit{Values}}}
\newcommand{\Vars}{\ensuremath{\mathit{Vars}}}
\newcommand{\semax}[2]{\ensuremath{[\![#1]\!]_{#2}}}
\newcommand{\Semax}[2]{\ensuremath{\{\!\!\!\{#1\}\!\!\!\}_{#2}}}
\newcommand{\upp}{\ensuremath{\sqcup}}
\newcommand{\fix}[2]{\ensuremath{\mu #1.#2}}
\newcommand{\envleq}[2]{\ensuremath{#1 \leq #2}}
\newcommand{\supp}[1]{\ensuremath{\mathit{sup}(#1)}}
\newcommand{\noteq}{\ensuremath{\not=}}
\newcommand{\undef}{\ensuremath{\bot}}
\newcommand{\evaltotal}[5]
{\ensuremath{\evalin{#1}{#2} \downarrow_{#5} \evalin{#3}{#4}}}
\newcommand{\evalname}[4]
{\ensuremath{\evalin{#1}{#2} \downarrow_{\it name} \evalin{#3}{#4}}}

\begin{document}

{\footnotetext{This paper is dedicated to the memory of Professor
  Reiji Nakajima (1947-2008).}}

\maketitle

\begin{abstract}
We present natural semantics for acyclic as well as cyclic
call-by-need lambda calculi, which are proved equivalent to the
reduction semantics given by Ariola and Felleisen.  The natural
semantics are big-step and use global heaps, where evaluation is
suspended and memorized.  The reduction semantics are small-step and
evaluation is suspended and memorized locally in let-bindings.  Thus
two styles of formalization describe the call-by-need strategy from
different angles.

The natural semantics for the acyclic calculus is revised from the
previous presentation by Maraist et al. and its adequacy is ascribed
to its correspondence with the reduction semantics, which has been
proved equivalent to call-by-name by Ariola and Felleisen.  The natural
semantics for the cyclic calculus is inspired by that of Launchbury
and Sestoft and we state its adequacy using a denotational semantics
in the style of Launchbury; adequacy of the reduction semantics for
the cyclic calculus is in turn ascribed to its correspondence with the
natural semantics.
\end{abstract}

\section{Introduction}\label{sec:intro}
In \cite{bigstepforlazy} Launchbury studied a natural semantics for
a call-by-need lambda calculus with letrec. He showed the semantics 
adequate using  a denotational semantics. 
Sestoft later revised Launchbury's semantics \cite{lazymachine}.
The revised semantics correctly enforces variable hygiene. 
Moreover the $\alpha$-renaming strategy
of the revised semantics is demonstrated to be suitable in the light of
possible implementations with heap-based abstract machines.  

In \cite{thecbnlambda} Ariola and Felleisen studied an equational theory
for an acyclic (non-recursive) call-by-need lambda calculus. 
The calculus admits the standardization theorem, which gives rise to
a reduction semantics for the calculus. 
The call-by-need evaluator, induced by the theory, is proved
equivalent to the call-by-name evaluator of Plotkin
\cite{cbn_cbv_lambda};
as a result, the reduction semantics is shown to be adequate. Ariola
and Felleisen also presented 
a cyclic (recursive) call-by-need lambda calculus with letrec;
however the cyclic calculus has not been explored. For instance, 
to the best of our knowledge, it has
not been known if the calculus relates to call-by-name or if the 
standard reduction relation, obtained from the one-step reduction relation
and evaluation contexts, is adequate. 

The two styles of formalization, namely the natural semantics and the
reduction semantics, describe the operational semantics for
call-by-need from different angles. The natural semantics is big-step
and evaluation is suspended and memorized in a global heap. Sestoft's
semantics rigorously preserves binding structure, by 
performing $\alpha$-renaming when allocating fresh locations in a
heap. As he demonstrated by deriving abstract machines from the
natural semantics, this approach to variable hygiene has a natural
correspondence with possible concrete implementations of call-by-need. 
The reduction semantics is small-step and evaluation is suspended and
memorized locally in let-bindings. It assumes implicit
$\alpha$-conversions. In fact we could think implicit renaming in
the reduction semantics is an appropriate approach to variable hygiene,
since freshness conditions cannot be checked locally. In other words, 
the reduction semantics allows for step-wise local reasoning
of program behavior using evaluation contexts. 

Our work is motivated to bridge the two styles of formalization, both
of which we found interesting. Here are contributions of the paper:
\begin{itemize}
\item We present natural semantics for acyclic and cyclic
call-by-need lambda calculi, and prove them equivalent to the
corresponding reduction semantics given by  Ariola and Felleisen.
For the acyclic calculus we revise the natural semantics given in
\cite{cbnlambda} by correctly enforcing variable hygiene 
in the style of Sestoft\footnote{In \cite{cbnlambda}
equivalence of the natural semantics and reduction semantics is stated. 
The paper only mentions that the result is proved by simple induction 
on derivations in the natural semantics, but we did not find it 
``simple''.}; its adequacy is ascribed to its correspondence with the
reduction semantics, which has been proved equivalent to call-by-name
by Ariola and Felleisen. 
The natural semantics for the cyclic calculus is very much inspired 
by Sestoft's, hence by Launchbury's;
the main difference is that our semantics directly works with the
full lambda terms with letrec, whereas Sestoft's works with
the ``normalized'' lambda terms, where function arguments are only
variables,  by having a precompilation step.

\item We show the natural semantics for the cyclic calculus 
adequate by adapting Launchbury's denotational argument. 
As a consequence the reduction semantics for the cyclic
calculus is also shown to be adequate thanks to the equivalence of the two
semantics; to the best of our knowledge, this fact has not been shown
so far. 
\end{itemize}

\section{Call-by-need let calculus \acyclic}\label{sec:acyclic}
We first study the operational semantics for the acyclic
(non-recursive) calculus.  

\subsection{Syntax and Semantics}

\begin{figure}
\hspace*{\fill}{\small $\begin{array}{llcl}
{\it Expressions}\hspace*{2ex}& 
M,N &::=& x \mid \fun{x}{M} \mid \app{M}{N} \mid \letin{\be{x}{M}}{N}\\
{\it Values}&
V &::=& \fun{x}{M}\\
{\it Answers}& 
A &::=& V \mid \letin{\be{x}{M}}{A}\\
{\it Contexts}&
E &::=& [] \mid \app{E}{M} \mid 
\letin{\be{x}{M}}{E} \mid\letin{\be{x}{E}}{E'[x]}\\
{\it Heaps}&
\Psi,\Phi &::=& \rt \mid \Psi, \maps{x}{M}
\end{array}$}\vspace{1ex}\hspace*{\fill}
\caption{Syntax of \acyclic}\label{fig:syntax}
\end{figure}

\begin{figure}
{\small \begin{tabular}{ll}
$\beta_{{\it need}}$:
\hspace*{5ex}&\app{(\fun{x}{M})}{N} \needmapsto\ \letin{\be{x}{N}}{M}\\
{\it lift}:
&\app{(\letin{\be{x}{M}}{A})}{N} \needmapsto\ \letin{\be{x}{M}}{\app{A}{N}}\\
{\it deref}:
&\letin{\be{x}{V}}{E[x]} \needmapsto\ \letin{\be{x}{V}}{E[V]}\\
{\it assoc}:
&\letin{\be{x}{(\letin{\be{y}{M}}{A})}}{E[x]} \needmapsto\
\letin{\be{y}{M}}{\letin{\be{x}{A}}{E[x]}}\vspace{.5ex}\\
\end{tabular}}\vspace{1ex}
\caption{Reduction semantics for \acyclic}\label{fig:red}
\end{figure}

\begin{figure}[t]
\[
{\small \begin{array}{c}
\mathit{Lambda}\\
\evalfresh{\Psi}{\fun{x}{M}}{\Psi}{\fun{x}{M}}{X}\vspace{1ex}\\
\mathit{Application}\\
\infer{
  \evalfresh{\Psi}{\app{M_1}{M_2}}{\Psi'}{V}{X}
}{
  \evalfresh{\Psi}{M_1}{\Phi}{\fun{x}{N}}{X}
  &\evalfresh{\Phi, \maps{x'}{M_2}}{\subst{N}{x'}{x}}{\Psi'}{V}{X}
  &x' ~\fresh
}
\vspace{1ex}\\
\mathit{Let}\\
\infer{
  \evalfresh{\Psi}{\letin{\be{x}{N}}{M}}{\Phi}{V}{X}
}{
  \evalfresh{\Psi,\maps{x'}{N}}{\subst{M}{x'}{x}}{\Phi}{V}{X}
  &x' ~\fresh
}
\vspace{1ex}\\
\mathit{Variable}\\
\infer{
  \evalfresh{\Psi,\maps{x}{M},\Phi}{x}{\Psi',\maps{x}{V},\Phi}{V}{X}
}{
  \evalfresh{\Psi}{M}{\Psi'}{V}{X \cup \{x\} \cup \dom {\Phi}}
}
\end{array}}
\]
\caption{Natural semantics for \acyclic}\label{fig:eval}
\end{figure}

The syntax of the call-by-need let calculus \acyclic\ is defined in
figure~\ref{fig:syntax}. The reduction and natural semantics are given
in figures \ref{fig:red} and \ref{fig:eval} respectively. The
metavariable $X$ ranges over sets of variables. 
The notation \rt\ denotes an empty sequence. 
The notation \dom{\Psi} denotes the domain of $\Psi$, namely
\dom{\rt} = $\emptyset$ and 
\dom{x_1 \mapsto M_1, \ldots, x_n \mapsto M_n} = 
$\{x_1, \ldots, x_n\}$.
The notation \subst{M}{x'}{x}
denotes substitution of $x'$ for free occurrences of $x$ in $M$. 
The notion of free variables is standard and is defined in figure
\ref{fig:freevar}. 
A {\it program} is a closed expression. 
We say an expression $M$ {\it (standard) reduces} to $N$, written $M \rar N$ if
$M = E[M']$ and $N = E[N']$ where $M' \needmapsto N'$.
We write $M \rarr N$ to denote that $M$ reduces to $N$ in zero or more steps,
i.e. \rarr\ is the reflexive and transitive closure of \rar. 

The reduction semantics is identical to 
the previous presentation by Ariola and Felleisen
\cite{thecbnlambda}. It works with $\alpha$-equivalence classes of
expressions. We assume 
all binding occurrences of variables in
a canonical representative of a class  use pairwise distinct names. In
particular, evaluation contexts and reduction rules are defined over
canonical representatives. Below we recall the reduction semantics
briefly.  The key rule is $\beta_{\mathit{need}}$, where application
reduces  to a {\sf let}-construct, thus suspending evaluation of the
argument.  Since {\it deref} only substitutes values for variables, 
$\beta_{\mathit{need}}$ also ensures that evaluation of an argument is
shared among all references to the argument in the function body. 
The administrative rules {\it lift} and {\it assoc} extend the scopes
of let-bound variables so that values surrounded by {\sf let}'s become
available without duplicating reducible expressions. 
The following lemma states that there exists at most one partitioning
of a program into a context and a redex, namely 
the unique-decomposition property. 
It is proved by induction on $M$.

\begin{lemma}\label{lemma:uniquecontext}
For any program $M$, $M$ is either an answer or there exist a unique
context $E$ and a redex $N$ such that $M = E[N]$.
\end{lemma}

The natural semantics is revised from that of Maraist et
al. \cite{cbnlambda}. It differs from the previous presentation in the
following two points. 
Firstly our semantics enforces variable hygiene correctly in the style
of Sestoft \cite{lazymachine} by keeping
track of variables which are temporarily deleted from heaps in 
{\it Variable} rule. This way, freshness 
conditions are locally checkable. Secondly our semantics works with 
the let-explicit calculus instead of the let-free one, hence has an
inference rule for the {\sf let}-construct; this makes it
smooth to extend our study of the acyclic calculus to the cyclic
calculus in the next section.  
As in \cite{cbnlambda}  the order of bindings in a heap
is significant.  
That is, re-ordering of bindings in a heap is not allowed. 
In particular in a heap $\maps{x_1}{M_1},\maps{x_2}{M_2},\ldots,
\maps{x_n}{M_n}$, an expression $M_i$ may contain as free
variables  only $x_1, \ldots, x_{i-1}$. This explains why 
it is safe to remove the bindings on the right in {\it Variable}
rule: $\Phi$ is not in the scope of $M$.
The natural semantics does not assume implicit $\alpha$-renaming,
but works with (raw) expressions. 
We may write \evalin{}{M} to denote \evalin{\rt}{M}.

\begin{figure}
\[
{\small \begin{array}{lcl}
\FVs{x} &=& \{x\}\\
\FVs{\fun{x}{M}} &=& \FVs{M}\backslash \{x\}\\
\FVs{\app{M}{N}} &=& \FVs{M} \cup \FVs{N}\\
\FVs{\letin{\be{x}{M}}{N}} &=& \FVs{M} \cup (\FVs{N}\backslash \{x\})
\end{array}}
\]
\caption{Free variables}\label{fig:freevar}
\end{figure}

A {\it configuration} is a pair \evalin{\Psi}{M} of a heap and an
expression.  
A configuration 
\evalin{\maps{x_1}{M_1},\ldots, \maps{x_n}{M_n}}{N}
is closed if 
$\FVs{N} \subseteq \{x_1, \ldots, x_n\}$, and
$\FVs{M_i} \subseteq \{x_1, \ldots, x_{i-1}\}$ 
for any $i$ in $1,\ldots, n$.
Borrowing from Sestoft's nomenclature \cite{lazymachine}, 
we say a configuration 
\evalin{\maps{x_1}{M_1},\ldots, \maps{x_n}{M_n}}{N}
is $X$-{\it good} if
$x_1, \ldots, x_n$ are pairwise distinctly named
and $\{x_1, \ldots, x_n\}$ and $X$ are disjoint. 
The judgment \evalfresh{\Psi}{M}{\Phi}{V}{X} is promising if 
\evalin{\Psi}{M} is closed and $X$-good. 

\medskip

Since derivations in the natural semantics only allocate fresh
variables in a heap
and substitute fresh variables for variables in expressions, 
a derivation of a promising judgment is promising everywhere. The
following lemma is proved by induction on the derivation of 
\evalfresh{\Psi}{M}{\Phi}{V}{X}.

\begin{lemma}\label{lemma:wf_inv}
If \evalin{\Psi}{M} is closed and $X$-good and the judgment 
\evalfresh{\Psi}{M}{\Phi}{V}{X} has a derivation,
then \evalin{\Phi}{V} is closed and $X$-good, and $\dom{\Psi} \subseteq
\dom{\Phi}$, and every judgment in the derivation is promising. 
\end{lemma}

Lemma~\ref{lemma:wf_inv} shows  the natural semantics preserves
binding structure in the absence of  implicit $\alpha$-renaming. 
Since the malloc function returns fresh locations in a heap,
the natural semantics indeed relates to heap-based implementations of
call-by-need.

\paragraph{Example}
Figures \ref{fig:red_acyclic_ex} and~\ref{fig:eval_acyclic_ex}
present the reduction sequence and the derivation 
for the expression {\sf let \be{x}{\app{(\fun{y}{y})}{(\fun{y}{y})}} in x} 
respectively.

\begin{figure}[t]
\[
\begin{array}{l}
\letin{\be{x}{\app{(\fun{y}{y})}{(\fun{y}{y})}}}{x}\\
\rar
\letin{\be{x}{(\letin{\be{y}{\fun{y}{y}}}{y})}}{x}\\
\rar
\letin{\be{x}{(\letin{\be{y}{\fun{y}{y}}}{\fun{y'}{y'}})}}{x}\\
\rar
\letin{\be{y}{\fun{y}{y}}}{\letin{\be{x}{\fun{y'}{y'}}}{x}}\\
\rar
\letin{\be{y}{\fun{y}{y}}}{\letin{\be{x}{\fun{y'}{y'}}}{\fun{y''}{y''}}}
\end{array}
\]
\caption{The reduction sequence for {\letin{\be{x}{\app{(\fun{y}{y})}{(\fun{y}{y})}}}{x}}}\label{fig:red_acyclic_ex}
\end{figure}

\begin{figure}[t]
\hspace*{\fill}\infer{
\evalfresh{}{\letin{\be{x}{\app{(\fun{y}{y})}{(\fun{y}{y})}}}{x}}
     {\maps{y'}{\fun{y}{y}},\maps{x'}{\fun{y}{y}}}{\fun{y}{y}}{\emp}
}{
  \infer{
    \evalfresh{\maps{x'}{\app{(\fun{y}{y})}{(\fun{y}{y})}}}{x'}
         {\maps{y'}{\fun{y}{y}},\maps{x'}{\fun{y}{y}}}{\fun{y}{y}}{\emp}
  }{
    \infer{
      \evalfresh{}{\app{(\fun{y}{y})}{(\fun{y}{y})}}{\maps{y'}{\fun{y}{y}}}{\fun{y}{y}}{\{x'\}}
    }{
      \evalfresh{}{\fun{y}{y}}{}{\fun{y}{y}}{\{x'\}}
      &\infer{
         \evalfresh{\maps{y'}{\fun{y}{y}}}{y'}
              {\maps{y'}{\fun{y}{y}}}{\fun{y}{y}}{\{x'\}}
       }{
         \evalfresh{}{\fun{y}{y}}{}{\fun{y}{y}}{\{x',y'\}}
       }
   }
  }
}\hspace*{\fill}\vspace{1ex}
\caption{The derivation for {\letin{\be{x}{\app{(\fun{y}{y})}{(\fun{y}{y})}}}{x}}}\label{fig:eval_acyclic_ex}
\end{figure}

\subsection{Equivalence of the two semantics}

\begin{figure}[t]
\hspace*{\fill}
{\small \begin{tabular}{c}
$\begin{array}{llcl}
\mathit{Frames}& F &::=& 
\app{[]}{M} \mid \letin{\be{x}{M}}{[]} \mid \letin{\be{x}{[]}}{E[x]}\\
\mathit{Structured~heaps}\hspace*{3ex}& \Sigma &::=& \rt \mid \Sigma, F\\
\mathit{Let's}& \Theta &::=&  \rt \mid \Theta,\letin{\be{x}{M}}{[]}\\
\end{array}$\vspace{3ex}\\
$\mathit{Lam}$\\
\infer{
  \Eval{\Sigma}{\fun{x}{M}}{\Sigma}{\fun{x}{M}}
}{}
\vspace{.5ex}\\
$\mathit{App}$\\
\infer{
  \Eval{\Sigma}{\app{M_1}{M_2}}{\Sigma_2}{V}
}{
  \Eval{\Sigma,\app{[]}{M_2}}{M_1}{\Sigma_1,\app{[]}{M_2},\Theta}{\fun{x}{N}}
  &\Eval{\Sigma_1,\Theta,\letin{\be{x'}{M_2}}{[]}}{\subst{N}{x'}{x}}{\Sigma_2}{V}
  &x' ~\fresh
}
\vspace{.5ex}\\
$\mathit{Letin}$\\
\infer{
  \Eval{\Sigma}{\letin{\be{x}{N}}{M}}{\Sigma'}{V}
}{
  \Eval{\Sigma,\letin{\be{x'}{N}}{[]}}{\subst{M}{x'}{x}}{\Sigma'}{V}
  &x' ~\fresh
}
\vspace{.5ex}\\
$\mathit{Var}$\\
\infer{
 \Eval{\Sigma,\letin{\be{x}{M}}{[]},\Sigma_1}{x}
      {\Sigma_2,\Theta,\letin{\be{x}{V}}{[]}, \Sigma_1}{V}
}{
 \Eval{\Sigma,\letin{\be{x}{[]}}{\Sigma_1[x]}}{M}
      {\Sigma_2,\letin{\be{x}{[]}}{\Sigma_1[x]},\Theta}{V}
}
\end{tabular}}\vspace{1ex}\hspace*{\fill}
\caption{Instrumented natural semantics for \acyclic}\label{fig:str_eval}
\end{figure}

The idea underlying our proof is derived from observing the following gap 
between the two semantics: 
\begin{itemize}
\item In the reduction semantics heaps are first allocated locally, 
then are globalized as much as necessary by applying 
{\it lift} or {\it assoc} afterwards 
to dereference computed values. Besides, the redex is focused
implicitly in the sense that the semantics does not specify how to
build evaluation contexts, but rather relies on the
unique-decomposition property.

\item In the natural semantics there is a single global heap. 
The redex is focused explicitly by applying inference rules,
thus decomposing evaluation contexts. 

\end{itemize}
To facilitate reconstructing reduction sequences from derivations by
bridging the above gap, our proof introduces an instrumented
natural semantics,  
defined in figure~\ref{fig:str_eval}, as an intermediary step. The
instrumented natural semantics uses {\it structured heaps}
$\Sigma$,  which are sequences of {\it frames} $F$. Intuitively
structured heaps are sequenced evaluation contexts.

The notation \LBVs{\Sigma} denotes the set of variables let-bound in frames 
of $\Sigma$. Or:
\[
\begin{array}{rcl}
\LBVs{\rt} &=& \emptyset\\
\LBVs{\Sigma,\app{[]}{M}} &=& \LBVs{\Sigma}\\
\LBVs{\Sigma,\letin{\be{x}{M}}{[]}} &=& \LBVs{\Sigma} \cup \{x\}\\
\LBVs{\Sigma,\letin{\be{x}{[]}}{M}} &=& \LBVs{\Sigma}\\
\end{array}
\]
A structured heap $\Sigma$ is well-formed if it is an empty sequence, 
or else $\Sigma = \Sigma',F$ and $\Sigma'$ is well-formed
and one of the following conditions holds:
\begin{enumerate}
\item $F = \app{[]}{M}$ and $\FVs{M} \subseteq \LBVs{\Sigma'}$
\item $F = \letin{\be{x}{M}}{[]}$ and 
$\FVs{M} \subseteq \LBVs{\Sigma'}$ and 
$x$ is distinct from
any of \LBVs{\Sigma'}
\item $F = \letin{\be{x}{[]}}{M}$ and 
$\FVs{M} \subseteq \LBVs{\Sigma'} \cup \{x\}$ and 
$x$ is distinct from any of \LBVs{\Sigma'}.
\end{enumerate}
A structured configuration \evalin{\Sigma}{M} is well-formed 
if $\Sigma$ is well-formed 
and $\FVs{M} \subseteq \LBVs{\Sigma}$. 

We map structured configurations to expressions 
by defining translation \comp{\cdot} 
from structured heaps to evaluation contexts: \vspace{.5ex}\\
\hspace*{\fill}$\begin{array}{rclrcl}
\comp{\rt} &=& []~~~~~&\comp{\Sigma,F} &=& \comp{\Sigma}[F]
\end{array}$\hspace*{\fill}\vspace{1ex}\\
We may identify $\Sigma$ with \comp{\Sigma} when there should be no
confusion, thus write $\Sigma[M]$ to denote $\comp{\Sigma}[M]$. 
A (raw) expression $\Sigma[M]$ is not necessarily a canonical
representative of an $\alpha$-equivalence class. The following lemma is
proved by induction on the structure of $\Sigma$.

\begin{lemma}\label{lamme:wfstrcnf_is_prg}
If \evalin{\Sigma}{M} is well-formed, then $\comp{\Sigma}[M]$ is a program. 
\end{lemma}

\medskip

Let's look at the inference rules in figure~\ref{fig:str_eval}.
{\it Lam} and {\it Letin} are self-explanatory.
When evaluating function expression $M_1$ in {\it App},
the rule pushes into the heap the frame \app{[]}{M_2}, which is popped when
evaluating function body $N$. 
Notice that the trailing frames to \app{[]}{M_2} in the result heap of the
left hypothesis is $\Theta$, which 
suggests $M_1$ reduces to an answer $\Theta[\fun{x}{N}]$.
This will be proved in Proposition~\ref{prop:Eval_then_red}.
Also, observe the order between $\Theta$ and \letin{\be{x'}{M_2}}{[]} in
the right hypothesis, 
where let-lifting is performed implicitly. 
When evaluating variable $x$ in {\it Var}, the rule pushes 
the ``continuation'' \letin{\be{x}{[]}}{\Sigma_1[x]}
into the heap. Again, observe the order between $\Theta$
and \letin{\be{x}{V}}{[]} in the result heap of the consequence, 
where let-association is implicitly performed.
It should be noted that Ariola and Felleisen already observed
that Launchbury's formalization has hidden flattening of a heap
in his {\it Variable} rule, which amounts to applying {\it assoc}~\cite{thecbnlambda}. 

\begin{figure}
\[
{\small \begin{array}{c}
\infer{
  \Eval{\letin{\be{x'}{[]}}{x'},\letin{\be{y'}{\fun{y}{y}}}{[]}}{y'}
       {\letin{\be{x'}{[]}}{x'},\letin{\be{y'}{\fun{y}{y}}}{[]}}
       {\fun{y}{y}}~~~~~(*)
}{
  \Eval{\letin{\be{x'}{[]}}{x'},\letin{\be{y'}{[]}}{y'}}{\fun{y}{y}}
       {\letin{\be{x'}{[]}}{x'},\letin{\be{y'}{[]}}{y'}}{\fun{y}{y}}
}
\\[1.5ex]
\infer{
\Eval{}{\letin{\be{x}{\app{(\fun{y}{y})}{(\fun{y}{y})}}}{x}}
       {\letin{\be{y'}{\fun{y}{y}}}{[]},\letin{\be{x'}{\fun{y}{y}}}{[]}}
       {\fun{y}{y}}
}{
  \infer{
    \Eval{\letin{\be{x'}{\app{(\fun{y}{y})}{(\fun{y}{y})}}}{[]}}{x'}
         {\letin{\be{y'}{\fun{y}{y}}}{[]},\letin{\be{x'}{\fun{y}{y}}}{[]}}
         {\fun{y}{y}}
  }{
    \infer{
      \Eval{\letin{\be{x'}{[]}}{x'}}{\app{(\fun{y}{y})}{(\fun{y}{y})}}
           {\letin{\be{x'}{[]}}{x'},\letin{\be{y'}{\fun{y}{y}}}{[]}}
           {\fun{y}{y}}              
    }{
      \deduce{
       \Eval{\letin{\be{x'}{[]}}{x'},\letin{\be{y'}{\fun{y}{y}}}{[]}}{y'}
            {\letin{\be{x'}{[]}}{x'},\letin{\be{y'}{\fun{y}{y}}}{[]}}
              {\fun{y}{y}}~~~~~(*)              
     }{
      \Eval{\letin{\be{x'}{[]}}{x'},\app{[]}{(\fun{y}{y})}}{\fun{y}{y}}
           {\letin{\be{x'}{[]}}{x'},\app{[]}{(\fun{y}{y})}}{\fun{y}{y}}
  }
   }
}
}
\end{array}}
\]
\caption{The derivation in the instrumented natural semantics for {\letin{\be{x}{\app{(\fun{y}{y})}{(\fun{y}{y})}}}{x}}}\label{fig:streval_acyclic_ex}
\end{figure}

\begin{lemma}\label{lemma:Eval_wfheap}
If \evalin{\Sigma}{M} is well-formed and 
\Eval{\Sigma}{M}{\Sigma'}{V},
then \evalin{\Sigma'}{V} is well-formed.
\end{lemma}
\begin{proof}
By induction on the derivation of \Eval{\Sigma}{M}{\Sigma'}{V}.
\end{proof}

Simple induction proves the instrumented natural semantics correct
with respect to the reduction semantics. 
\begin{proposition}\label{prop:Eval_then_red}
If \evalin{\Sigma}{M} is well-formed and \Eval{\Sigma}{M}{\Sigma'}{V}, 
then $\Sigma[M] \rarr \Sigma'[V]$.
\end{proposition}
\begin{proof}
By induction on the derivation of \Eval{\Sigma}{M}{\Sigma'}{V} with case analysis
on the last rule used.\\
- The cases of {\it Lam} and {\it Letin} are obvious.\\
- The case of {\it App}. Suppose 
we deduce \Eval{\Sigma}{\app{M_1}{M_2}}{\Sigma_2}{V} from
\Eval{\Sigma,\app{[]}{M_2}}{M_1}{\Sigma_1,\app{[]}{M_2},\Theta}{\fun{x}{N}}
and \Eval{\Sigma_1,\Theta,\letin{\be{x'}{M_2}}{[]}}{\subst{N}{x'}{x}}{\Sigma_2}{V}.
Then we have:\\
\hspace*{3ex}$\begin{array}{l}
\Sigma[\app{M_1}{M_2}]\\ 
\rarr \Sigma_1[\app{(\Theta[\fun{x}{N}])}{M_2}]~~~\mathrm{by~ind.~hyp.}\\
\rarr \Sigma_1[\Theta[\app{(\fun{x}{N})}{M_2}]]~~~\mathrm{by}~\mathit{lift}\\
\rar \Sigma_1[\Theta[\letin{\be{x'}{M_2}}{\subst{N}{x'}{x}}]]~~~\mathrm{by}~\beta_{\mathit{need}}\\
\rarr \Sigma_2[V]~~~\mathrm{by~ind.~hyp.}
\end{array}$\\
- The case of {\it Var}. Suppose we deduce
\Eval{\Sigma,\letin{\be{x}{M}}{[]},\Sigma_1}{x}
{\Sigma_2,\Theta,\letin{\be{x}{V}}{[]}, \Sigma_1}{V}
from
\Eval{\Sigma,\letin{\be{x}{[]}}{\Sigma_1[x]}}{M}
{\Sigma_2,\letin{\be{x}{[]}}{\Sigma_1[x]},\Theta}{V}.
Then we have:\\
\hspace*{3ex}$\begin{array}{l}
\Sigma[\letin{\be{x}{M}}{\Sigma_1[x]}]\\
\rarr \Sigma_2[\letin{\be{x}{\Theta[V]}}{\Sigma_1[x]}]~~~\mathrm{by~ind.~hyp.}\\
\rarr \Sigma_2[\Theta[\letin{\be{x}{V}}{\Sigma_1[x]}]]~~~\mathrm{by}~\mathit{assoc}\\
\rar \Sigma_2[\Theta[\letin{\be{x}{V}}{\Sigma_1[V]}]]~~~\mathrm{by}~\mathit{deref} 
\end{array}$

\end{proof}

\medskip

We need to prove the original natural semantics 
in figure~\ref{fig:eval} correct with respect to the instrumented 
natural semantics.
This is mainly to check that in figure~\ref{fig:str_eval} frames are 
properly pushed and popped so that the pop operation never fails.
Below we define a preorder on structured heaps to 
state that structured heaps only ``grow'' during derivations. 

A preorder \hpleq{}{} on structured heaps is defined such that
\hpleq{F_1,\ldots, F_m}{F'_1, \ldots, F'_n} 
if there is an injection $\iota$ from $\{1, \ldots, m\}$ to $\{1, \ldots, n\}$ 
satisfying the following three conditions:
\begin{enumerate}
\item  if $i < j$ then $\iota(i) < \iota(j)$ 
\item for all $i$ in $\{1, \ldots, m\}$, either $F_i = F'_{\iota(i)}$ or else
$F_i = \letin{\be{x}{M}}{[]}$ and 
$F'_{\iota(i)} = \letin{\be{x}{N}}{[]}$ for some $x$, $M$ and $N$
\item for all $i$ in $\{1, \ldots, n\}\backslash \ran{\iota}$,
$F'_i = \letin{\be{x}{M}}{[]}$ for some $x$ and $M$, where 
\ran{\iota} denotes the range of $\iota$ and 
$\{1, \ldots, n\}\backslash \ran{\iota}$ denotes set subtraction.
\end{enumerate}
It is easy to check that \hpleq{}{} is a preorder.

\begin{lemma}\label{lemma:Eval_heap_grow}
If \evalin{\Sigma}{M} is well-formed and \Eval{\Sigma}{M}{\Sigma'}{V}, 
then \hpleq{\Sigma}{\Sigma'}.
\end{lemma}
\begin{proof}
By induction on the derivation of \Eval{\Sigma}{M}{\Sigma'}{V}.
We use the fact that if \hpleq{\Sigma}{\Sigma'} and \hpleq{\Sigma',\Theta}{\Sigma''},
then \hpleq{\Sigma}{\Sigma''}.
\end{proof}

We define translation \decomp{\cdot} from structured heaps to
(ordinary) heaps  by collecting {\sf let}-frames as follows:
\[
\begin{array}{rcl}
\decomp{\rt} &=& \rt\\
\decomp{\Sigma,\app{[]}{M}} &=& \decomp{\Sigma}\\
\decomp{\Sigma,\letin{\be{x}{M}}{[]}}&=& \decomp{\Sigma},x \mapsto M\\
\decomp{\Sigma,\letin{\be{x}{[]}}{M}}&=& \decomp{\Sigma}
\end{array}
\]

\begin{proposition}\label{prop:eval_then_Eval}
If \evalin{\Psi}{M} is closed and $X$-good 
and \evalfresh{\Psi}{M}{\Phi}{V}{X}, 
then for any $\Sigma$ such that $\decomp{\Sigma} = \Psi$ 
and \evalin{\Sigma}{M} 
is well-formed, \Eval{\Sigma}{M}{\Sigma'}{V} and $\decomp{\Sigma'} = \Phi$.
\end{proposition}
\begin{proof}
By induction on the derivation of \evalfresh{\Psi}{M}{\Phi}{V}{X} 
with case analysis on
the last rule used.\\
- The cases of {\it Lambda} and {\it Let} are obvious.\\
- The case of {\it Application}. Suppose $M = \app{M_1}{M_2}$ and we deduce
\evalfresh{\Psi}{\app{M_1}{M_2}}{\Psi'}{V}{X}
from \evalfresh{\Psi}{M_1}{\Phi}{\fun{x}{N}}{X}
and \evalfresh{\Phi, \maps{x'}{M_2}}{\subst{N}{x'}{x}}{\Psi'}{V}{X}.
Suppose $\decomp{\Sigma} = \Psi$ and 
\evalin{\Sigma}{\app{M_1}{M_2}} is well-formed. 
By ind. hyp. and Lemma~\ref{lemma:Eval_wfheap} and~\ref{lemma:Eval_heap_grow},
\Eval{\Sigma,\app{[]}{M_2}}{M_1}{\Sigma_1,\app{[]}{M_2},\Theta}{\fun{x}{N}}
and $\decomp{\Sigma_1,\app{[]}{M_2},\Theta} = \Phi$ 
and \evalin{\Sigma_1,\app{[]}{M_2},\Theta}{\fun{x}{N}} is well-formed. 
By ind. hyp., \\
\Eval{\Sigma_1,\Theta,\letin{\be{x'}{M_2}}{[]}}{\subst{N}{x'}{x}}{\Sigma_2}{V}
and $\decomp{\Sigma_2} = \Psi'$.\\
- The case of {\it Variable}. Suppose $M = x$ and we 
deduce 
\evalfresh{\Psi,\maps{x}{N},\Phi}{x}{\Psi',\maps{x}{V},\Phi}{V}{X}
from \evalfresh{\Psi}{N}{\Psi'}{V}{X \cup \{x\} \cup \dom{\Phi}}. 
Let $\Sigma = \Sigma_1,\letin{\be{x}{N}}{[]}, \Sigma_2$ with
$\decomp{\Sigma_1} = \Psi$ and $\decomp{\Sigma_2} = \Phi$
and \evalin{\Sigma}{x} well-formed. 
By ind. hyp. and Lemma~\ref{lemma:Eval_heap_grow}, 
\Eval{\Sigma_1,\letin{\be{x}{[]}}{\Sigma_2[x]}}{N}
{\Sigma_3,\letin{\be{x}{[]}}{\Sigma_2[x]},\Theta}{V} with
$\decomp{\Sigma_3,\letin{\be{x}{[]}}{\Sigma_2[x]},\Theta} = \Psi'$.
Thus we deduce \Eval{\Sigma}{x}{\Sigma_3,\Theta,\letin{\be{x}{V}}{[]},\Sigma_2}{V}.
\end{proof}

\medskip

We prove the reduction semantics correct with respect to the
natural semantics 
without going through the instrumented natural semantics. 
We first prove three useful lemmas. 
Lemma~\ref{lemma:red_cxt} proves that irrelevant evaluation contexts
are replaceable. 
It lets us prove Lemma~\ref{lemma:red_app} and~\ref{lemma:red_let}.
The former proves that reductions  
at the function position inside application can be recast outside the
application.  
The latter proves that local reductions inside a let-binding can be 
recast as top-level reductions. 
We use the notation $M \rarr^n N$ to denote that $M$ reduces into $N$
in $n$ steps.

\begin{lemma}\label{lemma:red_cxt}
For any $\Theta$, $E$ and $x$ such that $\Theta[E[x]]$ is a program
and $x$ is not in \LBVs{E},
if $\Theta[E[x]] \rarr^n \Theta'[E[V]]$,
then for any $E'$ such that $\Theta[E'[x]]$ is a program and $x$ is
not in \LBVs{E'}, $\Theta[E'[x]] \rarr^n \Theta'[E'[V]]$.
\end{lemma}
\begin{proof}
By induction on $n$. 
Let $\Theta = \Theta_1,\letin{\be{x}{M}}{[]},\Theta_2$ with
$x$ not in \LBVs{\Theta_2}.
We perform case analysis on the possible reductions of $M$.\\
- The case where $M$ is an answer is easy.\\
- The case where  $M$ (one-step) reduces independently of the context
is immediate by induction. \\
- Suppose $M = E_1[x_1]$
and $x_1$ is not in \LBVs{E_1} and we have:\\
\hspace*{3ex}$\begin{array}{l}
\Theta_1[\letin{\be{x}{E_1[x_1]}}{\Theta_2[E[x]]}]
\rarr^{n_1} \Theta'_1[\letin{\be{x}{E_1[V_1]}}{\Theta_2[E[x]]}]
\rarr^{n_2} \Theta'[E[V]]
\end{array}$\\
Then by ind. hyp., we have:\\
\hspace*{3ex}$\begin{array}{l}
\Theta_1[\letin{\be{x}{E_1[x_1]}}{\Theta_2[E'[x]]}]
\rarr^{n_1} \Theta'_1[\letin{\be{x}{E_1[V_1]}}{\Theta_2[E'[x]]}]
\rarr^{n_2} \Theta'[E'[V]]
\end{array}$\\
\end{proof}

We introduce a notion of {\it rooted} reductions to identify
a particular intermediate step in reductions: 
a reduction $M \rar M'$ is 
$\beta_{\mathit{need}}$-{\it rooted with argument} $N$
if $M$ = $\Theta[\app{(\fun{x}{N'})}{N}]$ and 
$M' = \Theta[\letin{\be{x}{N}}{N'}]$. A reduction sequence 
$M \rarr M'$ {\it preserves a $\beta_{\mathit{need}}$-root with
argument $N$} if none of (one-step) reductions in the
sequence is $\beta_{\mathit{need}}$-rooted with argument $N$.
Intuitively, if $\Theta[\app{M}{N}] \rarr M'$ 
preserves a $\beta_{\mathit{need}}$-root with argument $N$,
then all the reductions only occur at $M$ or in the environment $\Theta$.

\begin{lemma}\label{lemma:red_app}
For any $\Theta$, $M$ and $N$ such that $\Theta[\app{M}{N}]$ is a program,
if $\Theta[\app{M}{N}] \rarr^n \Theta'[\app{V}{N}]$ and 
the reduction sequence preserves a $\beta_{\mathit{need}}$-root with
argument $N$, then $\Theta[M] \rarr^{n'} \Theta'[V]$ with $n' \leq n$.
\end{lemma}
\begin{proof}
By induction on $n$ with case analysis on the possible reductions of $M$.\\
- The case where $M$ is an answer is easy.\\
- The case where  $M$ reduces independently of the context
is immediate by induction. \\
- Suppose $M = E[x]$ and $x$ is not in \LBVs{E} and we have:\\ 
\hspace*{3ex}$\begin{array}{l}
\Theta[\app{(E[x])}{N}] 
\rarr^{n_1} \Theta_1[\app{(E[V])}{N}] \rarr^{n_2} \Theta'[\app{V}{N}]
\end{array}$\\
Then by Lemma~\ref{lemma:red_cxt} followed by ind. hyp., we have:\\
\hspace*{3ex}$\begin{array}{l}
\Theta[E[x]] 
\rarr^{n_1} \Theta_1[E[V]]
\rarr^{n'_2} \Theta'[V]
\end{array}$
where $n'_2 \leq n_2$.
\end{proof}

\begin{lemma}\label{lemma:red_let}
For any $\Theta, x, M$ and $E$ such that
$\Theta[\letin{\be{x}{M}}{E[x]}]$ is a program and $x$ is not in \LBVs{E},
if $\Theta[\letin{\be{x}{M}}{E[x]}] \rarr^n \Theta'[\letin{\be{x}{V}}{E[x]}]$
then $\Theta[M] \rarr^{n'} \Theta'[V]$ with $n' \leq n$.
\end{lemma}
\begin{proof}
By induction on $n$ with case analysis on the possible reductions
of $M$.\\
- The case where $M$ is an answer is easy.\\
- The case where  $M$ reduces independently of the context
is immediate by induction. \\
- Suppose $M = E'[x']$ and $x'$ is not in \LBVs{E'} and we have:\\
\hspace*{3ex}$\begin{array}{l}
\Theta[\letin{\be{x}{E'[x']}}{E[x]}]
\rarr^{n_1} \Theta_1[\letin{\be{x}{E'[V']}}{E[x]}]
\rarr^{n_2} \Theta'[\letin{\be{x}{V}}{E[x]}]
\end{array}$\\
Then by Lemma~\ref{lemma:red_cxt} followed by ind. hyp., we have:\\
\hspace*{3ex}$\begin{array}{l}
\Theta[E'[x']]
\rarr^{n_1} \Theta_1[E'[V']]
\rarr^{n'_2} \Theta'[V]
\end{array}$
where $n'_2 \leq n_2$.
\end{proof}

Now we are ready to prove the reduction semantics correct
with respect to the natural semantics, using the above three lemmas 
to have induction go through. 

\begin{proposition}\label{prop:red_then_eval}
For any program $M$, if  $M \rarr A$, 
then for any $X$, there exist $\Theta$ and $V$ such that $\Theta[V]$ and $A$
belong to the same $\alpha$-equivalence class and 
\evalfresh{}{M}{\decomp{\Theta}}{V}{X}.
\end{proposition}
\begin{proof}
Without loss of generality, we assume 
$\Theta[V]$ and $A$ are syntactically identical. 
We prove by induction on the length of the reductions of $M$. 
Let $M = \Theta'[M']$ with $M' \not= \letin{\be{x}{N'}}{N}$. We
perform case analysis on $M'$.\\
- The case of abstraction is obvious.\\
- The case of application. Suppose $M' = \app{M_1}{M_2}$ and we have:\\
\hspace*{3ex}$\begin{array}{l}
\Theta'[\app{M_1}{M_2}]
\rarr \Theta_1[\app{(\fun{x}{M_3})}{M_2}]
\rar \Theta_1[\letin{\be{x}{M_2}}{M_3}]
\rarr \Theta[V]
\end{array}$\\
By Lemma~\ref{lemma:red_app} and ind. hyp.,
\evalfresh{}{\Theta'[M_1]}{\decomp{\Theta_1}}{\fun{x}{M_3}}{X}.
By ind. hyp,
\evalfresh{}{\Theta_1[\letin{\be{x}{M_2}}{M_3}]}{\decomp{\Theta}}{V}{X}.
Thus we deduce 
\evalfresh{}{\Theta'[\app{M_1}{M_2}]}{\decomp{\Theta}}{V}{X}.\\
- The case of a variable. 
Suppose $M' = x$ and $\Theta' = \Theta_1,\letin{\be{x}{N}}{[]},\Theta_2$
and we have:\\
\hspace*{3ex}$\begin{array}{l}
\Theta_1[\letin{\be{x}{N}}{\Theta_2[x]}]
\rarr \Theta'_1[\letin{\be{x}{V}}{\Theta_2[x]}]
\rar \Theta'_1[\letin{\be{x}{V}}{\Theta_2[V]}]\\
\end{array}$\vspace{.5ex}\\
By Lemma~\ref{lemma:red_let} and ind. hyp., 
\evalfresh{}{\Theta_1[N]}{\decomp{\Theta'_1}}{V}
{X\cup \{x\} \cup \dom{\decomp{\Theta_2}}},
from which we deduce 
\evalfresh{}{\Theta'[x]}{\decomp{\Theta'_1,\letin{\be{x}{V}}{[]},\Theta_2}}{V}{X}.
\end{proof}

Collecting all propositions together, 
we prove the equivalence of the two semantics. 

\begin{theorem}
For any program $M$, the following two conditions hold:
\begin{enumerate}
\item if $M \rarr A$, then there exist $\Theta$ and $V$ such that 
$\Theta[V]$ and $A$ belong to the same $\alpha$-equivalence class and
\evalfresh{}{M}{\decomp{\Theta}}{V}{\emp}
\item if \evalfresh{}{M}{\Psi}{V}{\emp}, 
then $M \rarr \Theta[V]$ where $\decomp{\Theta} =  \Psi$.
\end{enumerate}
\end{theorem}
\begin{proof}
1: By Proposition~\ref{prop:red_then_eval}.
2: 
By Proposition~\ref{prop:eval_then_Eval}
and Lemma~\ref{lemma:Eval_heap_grow},  
\Eval{}{M}{\Theta}{V} with $\decomp{\Theta} = \Psi$.
By Proposition~\ref{prop:Eval_then_red}, $M \rarr \Theta[V]$.
\end{proof}

\section{Call-by-need letrec calculus \cyclic}\label{sec:cyclic}
In this section we extend the equivalence result to the cyclic
(recursive) calculus.  

\subsection{Syntax and semantics}
\begin{figure}
\hspace*{\fill}{\small$\begin{array}{llcl}
{\it Expressions}\hspace*{5ex}& 
M,N &::=& x \mid \fun{x}{M} \mid \app{M}{N} \mid \letrec{D}{M} \mid \error\\
{\it Definitions}&
D &::=& \rt \mid D, \be{x}{M}\\
{\it Values}& V &::=& \fun{x}{M} \mid \error\\
{\it Answers}& 
A &::=& V \mid \letrec{D}{A}\\
{\it Contexts}&
E &::=& [] \mid \app{E}{M} \mid 
\letrec{D}{E}\\ 
& &\mid& \letrec{\be{x}{E}, D}{E'[x]}\\
& &\mid& \letrec{\be{x'}{E}, \dps{x}{x'},D}{E'[x]}\\
{\it Dependencies}&
\dps{x}{x'} &::=& \be{x}{E[x']}\\ 
           &&\mid& \dps{x}{x''},\be{x''}{E[x']}\\
\end{array}$}\hspace*{\fill}\vspace{1ex}
\caption{Syntax of \cyclic}\label{fig:syntax_rec}
\end{figure}

The syntax of the call-by-need letrec calculus \cyclic\ is defined in
figure~\ref{fig:syntax_rec}. The reduction and natural semantics are defined
in figures \ref{fig:red_rec} and \ref{fig:eval_rec} respectively. 
No ordering among bindings in $D$ is assumed. 
Metavariables $\Psi$ and $\Phi$ range over finite mappings from
variables to expressions. 
Here we do not assume any ordering among bindings in heaps. 
In particular, a heap may contain cyclic structure such as
\evalin{\maps{x_1}{\fun{y}{\app{x_2}{y}}},\maps{x_2}{\fun{y}{\app{x_1}{y}}}}{}
and \evalin{\maps{x}{y},\maps{y}{x}}.
In the natural semantics, 
the notation \updateseq{\Psi}{x_i}{M_i} {i\in \{1,\ldots,n\}}
denotes mapping extension. Precisely,
\[
\updateseq{\Psi}{x_i}{M_i}{i\in \{1,\ldots,n\}}(x) = 
\left\{ \begin{array}{ll}
M_i & ~~{\rm when} ~x = x_i\ \mbox{for some} ~i ~\mathrm{in} ~1,\ldots,n\\
\Psi(x) & ~~{\rm otherwise}\\
\end{array}\right. 
\]
We write \update{\Psi}{x}{M} to denote a single extension of $\Psi$
with $M$ at $x$. In rule {\it Letrec} of figure~\ref{fig:eval_rec}, 
$M'_i$'s and $N'$ denote expressions
obtained from $M_i$'s and $N$ by substituting $x'_i$'s for $x_i$'s,
respectively.  
We may abbreviate \evalin{\Psi}{M} where $\Psi$ is an empty mapping,
i.e. the domain of $\Psi$ is empty, to \evalin{}{M}. We adapt the
definition of free variables in figure \ref{fig:freevar} for \cyclic\
by replacing the rule for {\sf let} with the
following rule:\vspace{1ex}\\
\hspace*{\fill}\scalebox{.95}{
$\begin{array}{l}
\FVs{\letrec{\be{x_1}{M_1}, \ldots, \be{x_n}{M_n}}{N}} \\
= (\FVs{M_1} \cup \ldots \cup \FVs{M_n} \cup \FVs{N}) \backslash 
\{x_1, \ldots, x_n \}
\end{array}$}\hspace*{\fill}\vspace{1ex}

\begin{figure}
\hspace*{\fill}{\small\begin{tabular}{ll}
$\beta_{\mathit{need}}:$
&\app{(\fun{x}{M})}{N} \needmapsto\ \letrec{\be{x}{N}}{M}\\
{\it lift}:
&\app{(\letrec{D}{A})}{N} \needmapsto\ \letrec{D}{\app{A}{N}}\\
{\it deref}:
&\letrec{\be{x}{V}, D}{E[x]} \needmapsto\
\letrec{\be{x}{V},D }{E[V]} \\
$\mathit{deref}_{\mathit{env}}:$
\hspace*{1ex}&\letrec{\dps{x}{x'}, \be{x'}{V},D}{E[x]} \needmapsto\
\letrec{\dps{x}{V}, \be{x'}{V}, D}{E[x]}\\
{\it assoc}:
&\letrec{\be{x}{(\letrec{D}{A})}, D'}{E[x]} \needmapsto\
\letrec{D, \be{x}{A}, D'}{E[x]}\\
$\mathit{assoc}_{\mathit{env}}:$
&\letrec{\be{x'}{(\letrec{D}{A})}, \dps{x}{x'}, D'}{E[x]} \needmapsto\\
&\letrec{D, \be{x'}{A}, \dps{x}{x'}, D'}{E[x]}\\
{\it error}:
&\letrec{\dps{x}{x}, D}{E[x]} \needmapsto\ \letrec{\dps{x}{\error},D}{E[x]}\\
$\mathit{error}_{\mathit{env}}:$
&\letrec{\dps{x'}{x'},\ensuremath{D'[x,x']} ,D}{E[x]} \needmapsto\ 
\letrec{\dps{x'}{\error},\ensuremath{D'[x,x']}, D}{E[x]}\\
$\mathit{error}_{\beta}:$
&\app{\error}{M} \needmapsto\ \error\\
\end{tabular}}\hspace*{\fill}\vspace{1ex}
\caption{Reduction semantics for \cyclic}\label{fig:red_rec}
\end{figure}

\begin{figure}[t]
\hspace*{\fill}{\small\begin{tabular}{c}
$\mathit{Value}$\\
\eval{\Psi}{V}{\Psi}{V}\vspace{1ex}\\
$\mathit{Application}$\\
\infer{
  \eval{\Psi}{\app{M_1}{M_2}}{\Psi'}{V}
}{
  \eval{\Psi}{M_1}{\Phi}{\fun{x}{N}}
  &\eval{\update{\Phi}{x'}{M_2}}{\subst{N}{x'}{x}}{\Psi'}{V}
  &x'~\fresh
}\vspace{1ex}\\
$\mathit{Variable}$\\
\infer{
  \eval{\Psi}{x}{\update{\Phi}{x}{V}}{V}
}{
  \eval{\update{\Psi}{x}{\error}}{\Psi(x)}{\Phi}{V}
}\vspace{1ex}\\
$\mathit{Letrec}$\\
\infer{
  \eval{\Psi}{\letrec{\be{x_1}{M_1},\ldots, \be{x_n}{M_n}}{N}}{\Phi}{V}
}{
  \eval{\updateseq{\Psi}{x'_i}{M'_i}{i\in \{1,...,n\}}}
       {N'}{\Phi}{V}
  &x'_1, \ldots, x'_n~\fresh
}\vspace{1ex}\\
$\mathit{Error}_{\beta}$\\
\infer{
  \eval{\Psi}{\app{M_1}{M_2}}{\Phi}{\error}
}{
  \eval{\Psi}{M_1}{\Phi}{\error}
}
\end{tabular}}\hspace*{\fill}\vspace{1ex}
\caption{Natural semantics for \cyclic}\label{fig:eval_rec}
\end{figure}

The reduction semantics is mostly identical to 
the previous presentation by Ariola and Felleisen \cite{thecbnlambda},
except that we elaborately deal with ``undefinedness'', which arises
due to direct cycles  
such as {\sf let rec} \be{x}{x} {\sf in} $M$. 
Undefinedness represents provable divergences. 
In our reduction semantics undefinedness, or 
black holes \error, are produced and propagated explicitly, in a spirit
similar to Wright and Felleisen's treatment of exceptions in a
reduction calculus \cite{wf-ic-94}.
Rules {\it error} and $\mathit{error}_{\mathit{env}}$ produce black holes.
Applying a black hole to an expression results in a black hole 
($\mathit{error}_{\beta}$). A value may be an abstraction or a black hole. 
Thus rules {\it lift}, {\it deref}, $\mathit{deref}_{\mathit{env}}$,
{\it assoc} 
and $\mathit{assoc}_{\mathit{env}}$ can be exercised to propagate
black holes.  
Explicit handling of black holes facilitates inductive reasoning. 
Again the reduction semantics works with $\alpha$-equivalence classes of
expressions. The following lemma states the unique-decomposition 
property for \cyclic\ and is proved by induction on $M$.

\begin{lemma}\label{lemma:uniquecontext_rec}
For any program $M$, $M$ is either an answer or there exist a unique
context $E$ and redex $N$ such that $M = E[N]$.
\end{lemma}

The natural semantics is very much inspired by Sestoft's
\cite{lazymachine}, hence by Launchbury's \cite{bigstepforlazy}. We
revise Sestoft's semantics in the following two 
points to draw a direct connection with the reduction
semantics. Firstly, in accordance with the reduction semantics, 
our natural semantics may return black holes. In {\it Variable} rule,
$x$ is bound to \error\ while the bound expression to $x$ is
evaluated. For instance, 
\eval{}{\letrec{\be{x}{x}}{x}}{\maps{x'}{\error}}{\error} is deduced
in our formulation.  Sestoft's formulation removes the binding of $x$
from the heap during its evaluation, 
thus evaluation involving direct cycles ``gets stuck'', 
i.e., no derivation is possible when direct cycles are encountered.
Since we do not remove bindings from heaps, freshness conditions are
locally checkable without extra variable tracking. 
Secondly, we do not precompile expressions into ``normalized''
ones. Our semantics works with full lambda expressions
with letrec, where function arguments may be any expressions, not only
variables.  

The notation \dom{\Psi} denotes the domain of $\Psi$.
A configuration \evalin{\Psi}{M} is closed 
if $\FVs{M} \subseteq \dom{\Psi}$, and
for any $x$ in \dom{\Psi}, $\FVs{\Psi(x)} \subseteq \dom{\Psi}$.

\paragraph{Example}
Figures \ref{fig:red_cyclic_ex} and \ref{fig:eval_cyclic_ex}
present the reduction sequence and the derivation 
for the expression \letrec{\be{x}{\app{f}{x}},\be{f}{\fun{y}{y}}}{x}
respectively. 
We deliberately chose a black hole producing expression to demonstrate
the difference of our formulation from 
Ariola and Felleisen's and Sestoft's.

\begin{figure}
\hspace*{\fill}$\begin{array}{l}
\letrec{\be{x}{\app{f}{x}},\be{f}{\fun{y}{y}}}{x}\\
\rar \letrec{\be{x}{\app{(\fun{y}{y})}{x}},\be{f}{\fun{y}{y}}}{x}\\
\rar \letrec{\be{x}{(\letrec{\be{y}{x}}{y})},\be{f}{\fun{y}{y}}}{x}\\
\rar \letrec{\be{x}{(\letrec{\be{y}{\error}}{y})},\be{f}{\fun{y}{y}}}{x}\\
\rar \letrec{\be{x}{(\letrec{\be{y}{\error}}{\error})},\be{f}{\fun{y}{y}}}{x}\\
\rar \letrec{\be{y}{\error},\be{x}{\error}, \be{f}{\fun{y}{y}}}{x}\\
\rar \letrec{\be{y}{\error},\be{x}{\error}, \be{f}{\fun{y}{y}}}{\error}\\
\end{array}$\hspace*{\fill}\vspace{1ex}
\caption{The reduction sequence for \letrec{\be{x}{\app{f}{x}},\be{f}{\fun{y}{y}}}{x}}\label{fig:red_cyclic_ex}
\end{figure}

\begin{figure}
\hspace*{\fill}
\scalebox{.8}{\infer{
  \eval{}{\letrec{\be{x}{\app{f}{x}},\be{f}{\fun{y}{y}}}{x}}
       {\maps{x'}{\error}, \maps{f'}{\fun{y}{y}},\maps{y'}{\error}}{\error}
}{
  \infer{
    \eval{\maps{x'}{\app{f'}{x'}}, \maps{f'}{\fun{y}{y}}}{x'}
         {\maps{x'}{\error}, \maps{f'}{\fun{y}{y}},\maps{y'}{\error}}{\error}
  }{
    \infer{
      \eval{\maps{x'}{\error}, \maps{f'}{\fun{y}{y}}}{\app{f'}{x'}}
           {\maps{x'}{\error}, \maps{f'}{\fun{y}{y}},\maps{y'}{\error}}{\error}
    }{
      \infer{
        \eval{\maps{x'}{\error}, \maps{f'}{\fun{y}{y}}}{f'}
             {\maps{x'}{\error}, \maps{f'}{\fun{y}{y}}}{\fun{y}{y}}
      }{
        \eval{\maps{x'}{\error}, \maps{f'}{\error}}{\fun{y}{y}}
             {\maps{x'}{\error}, \maps{f'}{\error}}{\fun{y}{y}}
      }
      &\infer{
        \eval{\maps{x'}{\error},\maps{f'}{\fun{y}{y}},\maps{y'}{x'}}{y'}
             {\maps{x'}{\error},\maps{f'}{\fun{y}{y}},\maps{y'}{\error}}{\error} 
      }{
        \infer{
          \eval{\maps{x'}{\error},\maps{f'}{\fun{y}{y}},\maps{y'}{\error}}{x'}
               {\maps{x'}{\error},\maps{f'}{\fun{y}{y}},\maps{y'}{\error}}{\error}
        }{
          \eval{\maps{x'}{\error},\maps{f'}{\fun{y}{y}},\maps{y'}{\error}}{\error}
               {\maps{x'}{\error},\maps{f'}{\fun{y}{y}},\maps{y'}{\error}}{\error}
        }
      }
    }
  }
}}\hspace*{\fill}\vspace{1ex}
\caption{The derivation for \letrec{\be{x}{\app{f}{x}},\be{f}{\fun{y}{y}}}{x}}\label{fig:eval_cyclic_ex}
\end{figure}

\subsection{Equivalence of the two semantics}

\begin{figure}[t]
\hspace*{\fill}{\small\begin{tabular}{c}
$\begin{array}{llcl}
{\it Frames}&
F &::=& \app{[]}{M} \mid \letrec{D}{[]} \mid \letrec{D_x,D}{E[x]}\\ 
{\it Structured ~heaps}~~&
\Sigma &::=& \rt \mid \Sigma, F \\
&D_x &::=& \be{x}{[]} \mid \dps{x}{x'},\be{x'}{[]}\\
\mathit{Letrec's}& \Theta &::=& \rt \mid \Theta,\letrec{D}{[]}
\end{array}$\vspace{2ex}\\
{\it Val}\\
\Eval{\Sigma}{V}{\Sigma}{V}
\vspace{1ex}\\
{\it App}\\
\infer{
  \Eval{\Sigma}{\app{M_1}{M_2}}{\Sigma_2}{V}
}{
  \deduce{
    \Eval{\Sigma_1,\Theta,\letrec{\be{x'}{M_2}}{[]}}{\subst{N}{x'}{x}}{\Sigma_2}{V}
  ~~x' ~\fresh
  }{
    \Eval{\Sigma,\app{[]}{M_2}}{M_1}{\Sigma_1,\app{[]}{M_2},\Theta}{\fun{x}{N}}
  }
}
\vspace{1ex}\\
{\it Letrecin}\\
\infer{
  \Eval{\Sigma}{\letrec{\be{x_1}{M_1},\ldots,\be{x_n}{M_n}}{N}}{\Sigma'}{V}
}{
  \Eval{\Sigma,\letrec{\be{x'_1}{M'_1},\ldots, \be{x'_n}{M'_n}}{[]}}{N'}{\Sigma'}{V}
  &x'_1, \ldots, x'_n~\fresh
}
\vspace{1ex}\\
{\it Var}\\
\infer{
  \Eval{\Sigma,\letrec{\be{x}{M},D}{[]},\Sigma_1}{x}
       {\Sigma',\letrec{\flatten{\Theta},\be{x}{V},D'}{[]},\Sigma_1}{V}
}{
  \Eval{\Sigma,\letrec{\be{x}{[]},D}{\Sigma_1[x]}}{M}
       {\Sigma',\letrec{\be{x}{[]},D'}{\Sigma_1[x]},\Theta}{V}
}
\vspace{1ex}\\
$\mathit{Var}_{\mathit{env}}$\\
\infer{
  \Eval{\Sigma,\letrec{\be{x}{M},D_{x'},D}{E[x']},\Sigma_1}{x}
{\Sigma',\letrec{\flatten{\Theta},\be{x}{V},D_{x'},D'}{E[x']},\Sigma_1}{V}
}{
  &\Eval{\Sigma,\letrec{\be{x}{[]},D_{x'}[\Sigma_1[x]],D}{E[x']}}{M}
        {\Sigma',\letrec{\be{x}{[]},D_{x'}[\Sigma_1[x]],D'}{E[x']},\Theta}{V}
}
\vspace{1ex}\\
$\mathit{Err}_{\mathit{var}}$\\
\infer{
  \Eval{\Sigma,\letrec{D,D_{x'}}{E[x']},\Sigma'}{x}
       {\Sigma,\letrec{D,D_{x'}}{E[x']},\Sigma'}{\error}
}{
  x \in D_{x'}
}
\vspace{1ex}\\
$\mathit{Err}_{\beta}$\\
\infer{
  \Eval{\Sigma}{\app{M_1}{M_2}}{\Sigma',\Theta}{\error}
}{
  \Eval{\Sigma,\app{[]}{M_2}}{M_1}{\Sigma',\app{[]}{M_2},\Theta}{\error}
}
\end{tabular}}\hspace*{\fill}\vspace{1ex}
\caption{Instrumented natural semantics for \cyclic}\label{fig:str_eval_rec}
\end{figure}

We prove equivalence of the two semantics for \cyclic\ in similar steps
to those for \acyclic, and use an instrumented natural semantics
defined in  
figure~\ref{fig:str_eval_rec}. 
The notation \flatten{\Theta} denotes the flattening of $\Theta$. Or:\\
\hspace*{3ex}$\begin{array}{rclrcl}
\flatten{\rt} &=& \rt
\hspace{4ex}&\flatten{\Theta,\letrec{D}{[]}} &=& \flatten{\Theta},D\\
\end{array}$\hspace*{\fill}\vspace{.5ex}\\
The notation $x \in D_{x'}$ denotes that $x$ is letrec-bound in 
$D_{x'}$, i.e., either \be{x}{[]} or \be{x}{M} is in $D_{x'}$.
In rule {\it Letrecin}, 
$M'_i$'s and $N'$ denote expressions
obtained from $M_i$'s and $N$ by substituting $x'_i$'s for $x_i$'s,
respectively.  

Here a frame may be \letrec{D}{[]} or \letrec{D_x,D}{E[x]}, 
instead of \letin{\be{x}{M}}{[]} or \letin{\be{x}{[]}}{E[x]}. 
We need to adjust the definitions of well-formedness for structured heaps
and structured configurations. 
The notation \LBVs{\Sigma} denotes the set of variables letrec-bound in frames 
of $\Sigma$. Or:
\[
\begin{array}{rcl}
\LBVs{\rt} &=& \emptyset\\
\LBVs{\Sigma,\app{[]}{M}} &=& \LBVs{\Sigma}\\
\LBVs{\Sigma,\letrec{D}{[]}} &=& 
\LBVs{\Sigma} \cup \LBVs{D}\\
\LBVs{\Sigma,\letrec{D, D_x}{M}}
&=& \LBVs{\Sigma} \cup \LBVs{D,D_x}\\
\LBVs{D, \be{x}{M}} &=& \LBVs{D} \cup \{x\}\\
\LBVs{D, \be{x}{[]}} &=& \LBVs{D} \cup \{x\}\\
\end{array}
\]
The notations \Body{F} and \Body{\Sigma} respectively denote 
the sets of expressions that $F$ and $\Sigma$ contain. Or:
\[
\begin{array}{rcl}
\Body{\app{[]}{M}} &=& \{M\}\\
\Body{\letrec{D}{[]}} &=& \Body{D}\\
\Body{\letrec{D,D_x}{M}} &=& \{M\} \cup \Body{D,D_x}\\
\Body{\rt} &=& \emptyset\\
\Body{D,\be{x}{M}} &=& \Body{D} \cup \{M\}\\
\Body{D,\be{x}{[]}} &=& \Body{D}\\
\Body{\Sigma,F} &=& \Body{\Sigma} \cup \Body{F}
\end{array}
\]
A structured heap $\Sigma$ is well-formed if it is an empty sequence,
or else $\Sigma = \Sigma',F$, and $\Sigma'$ is well-formed and 
one of the following conditions hold:
\begin{enumerate}
\item $F = \app{[]}{M}$  and $\FVs{M} \subseteq \LBVs{\Sigma}$
\item $F = \letrec{\be{x_1}{M_1},\ldots, \be{x_n}{M_n}}{[]}$ and 
$\FVs{M_i} \subseteq \LBVs{\Sigma}$ for all $i$'s, and 
$x_1, \ldots, x_n$ are pairwise distinctly named,
and all $x_i$'s are distinct from any of \LBVs{\Sigma'}

\item $F = \letrec{\be{x}{[]}, \be{x_1}{M_1},\ldots,  \be{x_n}{M_n}}{N}$ and 
$\FVs{N} \subseteq \LBVs{\Sigma}$ and $\FVs{M_i} \subseteq
  \LBVs{\Sigma}$ for all $i$'s, and  
$x, x_1, \ldots, x_n$ are pairwise distinctly named,
and all $x_i$'s and $x$ are distinct from any of \LBVs{\Sigma'},
\end{enumerate}
A structured configuration \evalin{\Sigma}{M} is well-formed 
if $\Sigma$ is well-formed and $\FVs{M} \subseteq \LBVs{\Sigma}$. 

We use the same definition as in the previous section 
for the translation \comp{\cdot} from structured heaps
to contexts:\\
\hspace*{\fill}$\begin{array}{rclrcl}
\comp{\rt} &=& []~~~~~&\comp{\Sigma,F} &=& \comp{\Sigma}[F]
\end{array}$\hspace*{\fill}\\
Again we may identify $\Sigma$ with \comp{\Sigma}, 
thus write $\Sigma[M]$ to denote $\comp{\Sigma}[M]$.
The following lemma is proved by induction on the structure of $\Sigma$.

\begin{lemma}
For any well-formed configuration \evalin{\Sigma}{M}, $\Sigma[M]$ is a program.
\end{lemma}

\medskip

Let's look at the inference rules in figure~\ref{fig:str_eval_rec}.
The first four rules are equivalent to the previous four rules in 
figure~\ref{fig:str_eval}. Whereas {\it Var} corresponds to the production
$\mathsf{let~rec} ~\be{x}{E},D$ $\mathsf{in}$ $E'[x]$ 
of evaluation contexts, $\mathit{Var}_{\mathit{env}}$ does to 
the production \letrec{\be{x'}{E},\dps{x}{x'},D}{E'[x]}. 
$\mathit{Err}_{\mathit{var}}$ mediates between the natural and reduction
semantics when a black hole is produced.
Indeed variables letrec-bound in $D_x$ correspond to 
variables bound to \error\ in a heap in the natural semantics. 
The instrumented natural semantics keeps the original expressions
bound to the variables  
to facilitate reconstructing reduction sequences from its derivations. 
${\mathit{Err}_{\beta}}$ is almost the same 
as the original rule $\mathit{\it Error}_{\beta}$ in figure
\ref{fig:eval_rec}. 

\begin{lemma}\label{lemma:Eval_wfheap_rec}
If \evalin{\Sigma}{M} is well-formed and \Eval{\Sigma}{M}{\Sigma'}{V},
then \evalin{\Sigma'}{V} is well-formed. 
\end{lemma}
\begin{proof}
By induction on the derivation of \eval{\Sigma}{M}{\Sigma'}{V}.
\end{proof}

\medskip

Easy induction proves the instrumented natural semantics 
correct with respect to the reduction semantics. 

\begin{proposition}\label{prop:Eval_then_red_rec}
If \evalin{\Sigma}{M} is well-formed and \Eval{\Sigma}{M}{\Sigma'}{V},
then $\Sigma[M] \rarr \Sigma'[V]$.
\end{proposition}
\begin{proof}
By induction on the derivation of \Eval{\Sigma}{M}{\Sigma'}{V}
with case analysis on the last rule used. \\
- The case of {\it Val} is obvious.\\
- The case of {\it App}. Suppose we deduce
\Eval{\Sigma}{\app{M_1}{M_2}}{\Sigma'}{V}
from 
\Eval{\Sigma,\app{[]}{M_2}}{M_1}{\Sigma_1,\app{[]}{M_2},\Theta}{\fun{x}{N}}
and 
\Eval{\Sigma_1,\Theta,\letrec{\be{x'}{M_2}}{[]}}{\subst{N}{x'}{x}}{\Sigma'}{V}.
Then we have:\\
\hspace*{3ex}$\begin{array}{l}
\Sigma[\app{M_1}{M_2}]\\
\rarr \Sigma_1[\app{(\Theta[\fun{x}{N}])}{M_2}]~~~\mathrm{by~ind.~hyp.}\\
\rarr \Sigma_1[\Theta[\app{(\fun{x}{N})}{M_2}]]~~~\mathrm{by}~\mathit{lift}\\
\rar \Sigma_1[\Theta[\letrec{\be{x'}{M_2}}{\subst{N}{x'}{x}}]]
~~~\mathrm{by}~\beta_{\mathit{need}}\\
\rarr \Sigma'[V]~~~\mathrm{by~ind.~hyp.}
\end{array}$\\
- The case of {\it Letrecin} is immediate by induction.\\
- The case of {\it Var}. Suppose we deduce
\Eval{\Sigma,\letrec{\be{x}{M},D}{[]},\Sigma_1}{x}
{\Sigma_2,\letrec{\flatten{\Theta},\be{x}{V},D'}{[]},\Sigma_1}{V}
from 
\Eval{\Sigma,\letrec{\be{x}{[]}, D}{\Sigma_1[x]}}{M}
{\Sigma_2,\letrec{\be{x}{[]}, D'}{\Sigma_1[x]},\Theta}{V}.
Then we have:\\
\hspace*{3ex}$\begin{array}{l}
\Sigma[\letrec{\be{x}{M}, D}{\Sigma_1[x]}]\\
\rarr \Sigma_2[\letrec{\be{x}{\Theta[V]},D'}{\Sigma_1[x]}]~~~\mathrm{by~ind.~hyp.}\\
\rarr \Sigma_2[\letrec{\flatten{\Theta}, \be{x}{V}, D'}{\Sigma_1[x]}]
~~~\mathrm{by}~\mathit{assoc}\\
\rar \Sigma_2[\letrec{\flatten{\Theta}, \be{x}{V},D'}{\Sigma_1[V]}]
~~~\mathrm{by}~\mathit{deref}\\
\end{array}$\vspace{.5ex}\\
- The case of $\mathit{Var}_{\mathit{env}}$ is similar to the above
{\it Var} case, where we use $\mathit{assoc}_{\mathit{env}}$
and $\mathit{deref}_{\mathit{env}}$ instead of {\it assoc} and {\it deref},
respectively. \\
- The case of $\mathit{Err}_{\mathit{var}}$ (1). Suppose $x=x'$ and we deduce
\Eval{\Sigma,\letrec{D,D_x}{E[x]},\Sigma'}{x}
{\Sigma,\letrec{D,D_x}{E[x]},\Sigma'}{\error}.
The side-condition $x \in D_x$ implies $ D_x[\Sigma'[x]] = \dps{x}{x}$. 
Thus we have 
$\Sigma[\letrec{D, D_{x}[\Sigma'[x]]}{E[x]}]
\rar \Sigma[\letrec{D, D_{x}[\Sigma'[\error]]}{E[x]}]$ by {\it error}. \\
- The case of $\mathit{Err}_{\mathit{var}}$ (2). 
Suppose $x \not= x'$ and we deduce
\Eval{\Sigma,\letrec{D, D_{x'}}{E[x']},\Sigma'}{x}
{\Sigma,\letrec{D, D_{x'}}{E[x']},\Sigma'}{\error}. 
Then $x \in D_{x'}$ implies $D_{x'}[\Sigma'[x]] = \dps{x'}{x},\dps{x}{x}$.
Thus we have
$\Sigma[\letrec{D, D_{x'}[\Sigma'[x]]}{E[x']}]
\rar \Sigma[\letrec{D, D_{x'}[\Sigma'[\error]]}{E[x']}]$ by 
$\mathit{error}_{\mathit{env}}$.\\
- The case of $\mathit{Err}_{\beta}$ is easy and similar to {\it App}.
\end{proof}

\medskip

Next we prove the instrumented natural semantics correct
with respect to the original natural semantics in figure~\ref{fig:eval_rec}.
Again this amounts to check that in the instrumented natural semantics
pushing and popping frames into heaps are properly balanced.
The proof is similar to the previous one 
for Proposition~\ref{prop:eval_then_Eval}, but we extend the preorder
\hpleq{}{} on structured heaps to take account of their cyclic structure. 

To define the preorder \hpleq{}{} on structured heaps, we use 
two auxiliary preorders. 
The preorder $\leq_{\mathcal{D}}$ on sequences of bindings is defined such
that $D \leq_{\mathcal{D}} D'$ if $\LBVs{D} \subseteq \LBVs{D'}$.
The preorder $\leq_{\mathcal{F}}$ on frames is the smallest reflexive and
transitive relation satisfying the condition that 
if $D \leq_{\mathcal{D}} D'$, then 
\fleq{\letrec{D_x,D}{E[x]}}{\letrec{D_x,D'}{E[x]}}
and \fleq{\letrec{D}{[]}}{\letrec{D'}{[]}}.
Then the preorder \hpleq{}{} on structured heaps is defined such that
\hpleq{F_1,\ldots, F_m}{F'_1, \ldots, F'_n} 
if there is an injection $\iota$ from $\{1, \ldots, m\}$ to $\{1, \ldots, n\}$ 
satisfying the following three conditions:
\begin{enumerate}
\item  if $i < j$ then $\iota(i) < \iota(j)$ 
\item for all $i$ in $\{1, \ldots, m\}$, \fleq{F_i}{F'_{\iota(i)}}
\item for all $i$ in $\{1, \ldots, n\}\backslash \ran{\iota}$,
$F'_i = \letrec{D}{[]}$ for some $D$. 
\end{enumerate}
It is easy to check that \hpleq{}{} is a preorder. 
The following lemma is proved by induction on the derivation of
\Eval{\Sigma}{M}{\Sigma'}{V}. 

\begin{lemma}\label{lemma:Eval_heap_grow_rec}
If \evalin{\Sigma}{M} is well-formed and \Eval{\Sigma}{M}{\Sigma'}{V}, 
then \hpleq{\Sigma}{\Sigma'}.
\end{lemma}

We define translation \decomp{\cdot} from structured heaps into 
sequences of bindings by:
\[
\begin{array}{rcl}
\decomp{\rt} &=& \rt\\
\decomp{\Sigma,\app{[]}{M}} &=& \decomp{\Sigma}\\
\decomp{\Sigma,\letrec{D}{[]}} &=& \decomp{\Sigma},D\\
\decomp{\Sigma,\letrec{D,D_x}{M}} &=& 
\decomp{\Sigma},D,\be{x_1}{\error},\ldots,\be{x_n}{\error}\\
\end{array}
\]
where $\LBVs{D_x} = \{x_1, \ldots, x_n\}$. 
We identify a sequence of bindings $D$ 
with a heap $\Psi$ such that
\LBVs{D} = \dom{\Psi}, and for all $x$ in \dom{\Psi}, 
$\Psi(x) = M$ iff $D$ contains $\be{x}{M}$. 
Thus \decomp{\Sigma} denotes a heap. 

We prove one basic result about the natural semantics:
Lemma~\ref{lemma:eval_heap_mono_rec} states that 
extending heaps with irrelevant bindings does not affect derivations
and is proved by routine induction. 
For mappings $\Psi, \Phi$ such that \dom{\Psi} and \dom{\Phi} are disjoint,
the notation \uni{\Psi}{\Phi} denotes their union, namely 
\dom{\uni{\Psi}{\Phi}} = $\dom{\Psi} \cup \dom{\Phi}$ and:\\
\hspace*{5ex}$(\uni{\Psi}{\Phi})(x) = 
\left\{ \begin{array}{ll}
\Psi(x) & ~~{\rm when} ~x \in \dom{\Psi} \\
\Phi(x) & ~~{\rm when} ~x \in \dom{\Phi} \\
\end{array}\right. $

\begin{lemma}\label{lemma:eval_heap_mono_rec}
For any $\Psi$, $\Psi'$, $\Phi$ and $M$
such that \dom{\Psi'} and \dom{\Phi} are disjoint 
and \evalin{\Psi}{M} and \evalin{\uni{\Psi}{\Psi'}}{M} are closed,
\eval{\Psi}{M}{\Phi}{V} iff \eval{\uni{\Psi}{\Psi'}}{M}{\uni{\Phi}{\Psi'}}{V}
and their derivations are of the same depth. 
\end{lemma}

\begin{proposition}\label{prop:eval_then_Eval_rec}
If \evalin{\Psi}{M} is closed and \eval{\Psi}{M}{\Phi}{V}, 
then for any $\Sigma$ such that 
$\decomp{\Sigma} = \Psi$ and \evalin{\Sigma}{M} is well-formed, 
\Eval{\Sigma}{M}{\Sigma'}{V} with $\decomp{\Sigma'} = \Phi$.
\end{proposition}
\begin{proof}
By induction on the depth of the derivation of \eval{\Psi}{M}{\Phi}{V} with case analysis 
on the last rule used. \\
- The case of {\it Value} is obvious.\\
- The case of {\it Application}. Suppose $\decomp{\Sigma} = \Psi$ 
and \evalin{\Sigma}{\app{M_1}{M_2}} is well-formed and we deduce
\eval{\Psi}{\app{M_1}{M_2}}{\Psi'}{V}
from 
\eval{\Psi}{M_1}{\Phi}{\fun{x}{N}}
and \eval{\update{\Phi}{x'}{M_2}}{\subst{N}{x'}{x}}{\Psi'}{V}.
By ind. hyp. and Lemma~\ref{lemma:Eval_heap_grow_rec},
\Eval{\Sigma,\app{[]}{M_2}}{M_1}{\Sigma_1,\app{[]}{M_2},\Theta}{\fun{x}{N}}.
with \decomp{\Sigma_1,\app{[]}{M_2},\Theta} = $\Phi$.
By Lemma~\ref{lemma:Eval_wfheap_rec}, 
\evalin{\Sigma_1,\app{[]}{M_2},\Theta}{\fun{x}{N}} is well-formed. 
By ind. hyp., 
\Eval{\Sigma_1,\Theta,\letrec{\be{x'}{M_2}}{[]}}{\subst{N}{x'}{x}}{\Sigma_2}{V}
with \decomp{\Sigma_2} = $\Psi'$.\\
- The cases of $\mathit{Error}_{\beta}$ and {\it Letrec} are immediate by induction. \\
- The case of {\it Variable}. Suppose we deduce
\eval{\Psi}{x}{\update{\Phi}{x}{V}}{V}
from
\eval{\update{\Psi}{x}{\error}}{\Psi(x)}{\Phi}{V}.
Suppose $\decomp{\Sigma} = \Psi$ and \evalin{\Sigma}{x} is well-formed.
There are three possible cases.\\ 
- - When $\Psi(x) = \error$ and 
$\Sigma = \Sigma_1, \letrec{D,D_{x'}}{E[x']},\Sigma_2$ with $x \in D_{x'}$.
Then we deduce \Eval{\Sigma}{x}{\Sigma}{\error} by ${\mathit{Err}_{\mathit{var}}}$.\\
- - When $\Psi(x) = N$ 
and $\Sigma = \Sigma_1,\letrec{\be{x}{N},D}{[]},\Sigma_2$. 
By ind. hyp. and Lemma~\ref{lemma:Eval_heap_grow_rec} and~\ref{lemma:eval_heap_mono_rec},
\Eval{\Sigma_1,\letrec{\be{x}{[]},D}{\Sigma_2[x]}}{N}
{\Sigma'_1,\letrec{\be{x}{[]},D'}{\Sigma_2[x]},\Theta}{V}
and \decomp{\Sigma'_1,\letrec{\be{x}{[]},D'}{\Sigma_2[x]},\Theta} is
the restriction of $\Phi$ to\\
\LBVs{\Sigma'_1,\letrec{\be{x}{[]},D'}{\Sigma_2[x]},\Theta}.
Hence by {\it Var} we deduce \\
\Eval{\Sigma_1,\letrec{\be{x}{N},D}{[]},\Sigma}{x}
{\Sigma'_1,\letrec{\flatten{\Theta},\be{x}{V},D'}{[]},\Sigma_2}{V}
and 
$\decomp{\Sigma'_1,\letrec{\flatten{\Theta},\be{x}{V},D'}{[]},\Sigma_2} = \update{\Phi}{x}{V}$.\\
- - The case where $\Psi(x) = N$ and 
$\Sigma = \Sigma_1,\letrec{\be{x}{N},D,D_{x'}}{E[x']},\Sigma_2$ is similar to 
the above case, except that we use $\mathit{Var}_{\mathit{env}}$ instead of {\it Var}.
\end{proof}

\medskip

We prove the reduction semantics correct with respect to the natural
semantics by proving three auxiliary results in
Lemma~\ref{lemma:red_letrec} and~\ref{lemma:red_app_rec}
and Corollary~\ref{coro:red_cxt_rec}, which respectively 
correspond to Lemma~\ref{lemma:red_let}, \ref{lemma:red_app}
and~\ref{lemma:red_cxt} for the acyclic case. 

We say a reduction sequence $M \rarr^n N$ is {\it autonomous}
if either $n = 0$,  or else the last step is reduced by rules other
than {\it assoc} or  
$\mathit{assoc}_{\mathit{env}}$. 
These two rules have particular behaviour in that 
they flatten nested {\sf letrec}'s on request outside;
we will restrict the use of the two rules by requiring a reduction sequence
to be autonomous. 
We write $M \rarR^n N$ to denote that $M$ reduces into $N$ in $n$-steps 
and the reduction sequence is autonomous. 
We may omit the suffix $n$ when it is irrelevant.

\begin{lemma}\label{lemma:red_letrec}
The following two conditions hold.
\begin{enumerate}
\item For any $\Theta$, $x$, $M$, $D$ and $E$ such that
$\Theta[\letrec{\be{x}{M},D}[E[x]]]$ is a program and
$x$ is not in \LBVs{E},
$\Theta[\letrec{\be{x}{M},D}[E[x]]] \rarR^n 
\Theta'[\letrec{\be{x}{A},D'}{E[x]}]$ iff\\
$\Theta[\letrec{\be{x}{\error},D}{M}] 
\rarr^n \Theta'[\letrec{\be{x}{\error},D'}{A}]$

\item For any $\Theta$, \dps{x_1}{x_m}, $M$, $D$ and $E$ such that
$\Theta[\letrec{\dps{x_1}{x_m},\be{x_m}{M},D}{E[x_1]}]$ is a program
and $x_1$ is not in \LBVs{E} and \LBVs{\dps{x_1}{x_m}} = 
$\{x_1, \ldots, x_{m-1}\}$,\\
$\Theta[\letrec{\dps{x_1}{x_m},\be{x_m}{M},D}{E[x_1]}]$ $\rarR^n
\Theta'[\letrec{\dps{x_1}{x_m},\be{x_m}{A},D'}{E[x_1]}]$ iff
$\Theta[\letrec{\be{x_1}{\error}, \ldots, \be{x_m}{\error},D}{M}] 
\rarr^n \Theta'[\letrec{\be{x_1}{\error}, \ldots, \be{x_m}{\error},D'}{A}]$.
\end{enumerate}
\end{lemma}
\begin{proof}
First we remark that the autonomy condition uniquely determines $n$
in the if case 
of both the conditions. 
We prove by simultaneous induction on the length of the reductions
with case analysis on the possible reductions.\\
- The case where $M$ is an answer is obvious.\\
- The case where  $M$ reduces independently of the context
is immediate by induction. \\
- The case where $M = E'[x']$ and $\Theta = \Theta_1,\letrec{\be{x'}{N},D_1}{[]},\Theta_2$.
We only prove the if case in 1. The other cases are similar. 
Suppose we have:\\
\hspace*{3ex}$\begin{array}{l}
\Theta_1[\letrec{\be{x'}{N},D_1}{\Theta_2[\letrec{\be{x}{E'[x']},D}[E[x]]]}]\\
\rarR^{n_1} \Theta'_1[\letrec{\be{x'}{\Theta_3[V]},D'_1}{\Theta_2[\letrec{\be{x}{E'[x']},D}[E[x]]]}]\\
\rarr^{n_2} \Theta'_1[\letrec{\be{x'}{V},\flatten{\Theta_3},D'_1}{\Theta_2[\letrec{\be{x}{E'[x']},D}[E[x]]]}]\\
\rar \Theta'_1[\letrec{\be{x'}{V},\flatten{\Theta_3},D'_1}{\Theta_2[\letrec{\be{x}{E'[V]},D}[E[x]]]}]\\
\rarR^{n_3} \Theta'[\letrec{\be{x}{A},D'}{E[x]}]
\end{array}$\\
By ind. hyp., 
$\Theta_1[\letrec{\be{x'}{\error},D_1}{N}]
\rar^{n_1} \Theta'_1[\letrec{\be{x'}{\error},D'_1}{\Theta_3[V]}]$.\\
Hence we have:\\
\hspace*{3ex}$\begin{array}{l}
\Theta_1[\letrec{\be{x'}{N},D_1}{\Theta_2[\letrec{\be{x}{\error},D}{E'[x']}]}]\\
\rarR^{n_1} \Theta'_1[\letrec{\be{x'}{\Theta_3[V]},D'_1}{\Theta_2[\letrec{\be{x}{\error},D}{E'[x']}]}]~~~\mathrm{by~ind.~hyp.}\comment{1:only~if}\\
\rarr^{n_2} \Theta'_1[\letrec{\be{x'}{V}, \flatten{\Theta_3},D'_1}{\Theta_2[\letrec{\be{x}{\error},D}{E'[x']}]}]~~~\mathrm{by}~\mathit{assoc}\\
\rar \Theta'_1[\letrec{\be{x'}{V}, \flatten{\Theta_3},D'_1}{\Theta_2[\letrec{\be{x}{\error},D}{E'[V]}]}]~~~\mathrm{by}~\mathit{deref}\\
\rarr^{n_3} \Theta'[\letrec{\be{x}{\error},D'}{A}]~~~\mathrm{by~ind.~hyp.}\comment{1:if}\\
\end{array}$\\
- The cases where $M = E'[x]$ in 1. and where $M = E'[x_i]$ for some $i$ in $1, \ldots, m$
in 2. are immediate by induction.\\
- The case where $M = E'[x']$ and $x'$ is in \LBVs{D} for the if case in 1.
Suppose we have:\\
\hspace*{3ex}$\begin{array}{l}
\Theta[\letrec{\be{x}{E'[x']},\be{x'}{N},D_1}{E[x]}]\\
\rarR^{n_1} \Theta_1[\letrec{\be{x}{E'[x']},\be{x'}{\Theta_2[V]},D'_1}{E[x]}]\\
\rarr^{n_2} \Theta_1[\letrec{\be{x}{E'[x']},\flatten{\Theta_2}, \be{x'}{V},D'_1}{E[x]}]\\
\rar \Theta_1[\letrec{\be{x}{E'[V]},\flatten{\Theta_2}, \be{x'}{V},D'_1}{E[x]}]\\
\rarR^{n_3} \Theta'[\letrec{\be{x}{A},D'}{E[x]}]
\end{array}$\\
By ind. hyp.\comment{2:if},
$\Theta[\letrec{\be{x}{\error}, \be{x'}{\error}, D_1}{N}] \rarr^{n_1} \Theta_1[\letrec{\be{x}{\error},\be{x'}{\error}, D'_1}{\Theta_2[V]}]$. Hence we have:\\
\hspace*{3ex}$\begin{array}{l}
\Theta[\letrec{\be{x}{\error}, \be{x'}{N},D_1}{E'[x']}]\\
\rarR^{n_1} \Theta_1[\letrec{\be{x}{\error}, \be{x'}{\Theta_2[V]},D'_1}{E'[x']}]~~~\mathrm{by~ind.~hyp.}\comment{1:only if}\\
\rarr^{n_2} \Theta_1[\letrec{\be{x}{\error},\flatten{\Theta_2},\be{x'}{V},D'_1}{E'[x']}]~~~\mathrm{by}~\mathit{assoc}\\
\rar \Theta_1[\letrec{\be{x}{\error},\flatten{\Theta_2},\be{x'}{V},D'_1}{E'[V]}]~~~\mathrm{by}~\mathit{deref}\\
\rarr \Theta'[\letrec{D'}{A}]~~~\mathrm{by~ind.~hyp.}
\end{array}$\\
- The cases where $M = E'[x']$ and $x'$ is in \LBVs{D} for the only if case in 1.
and the if and only if cases in 2. are similar to the above case. 
\end{proof}

\begin{corollary}\label{coro:red_cxt_rec}
For any $\Theta$, $E$ and $x$ such that 
$\Theta[E[x]]$ is a program and $x$ is not in \LBVs{E},
if $\Theta[E[x]] \rarr^n \Theta'[E[V]]$,
then for any $E'$ such that $\Theta[E'[x]]$ is a program and $x$ is
not in \LBVs{E'}, 
$\Theta[E'[x]] \rarr^n \Theta'[E'[V]]$.
\end{corollary}

We adapt the definition of rooted reductions in an obvious way by
replacing {\sf let} with {\sf let rec}.
A reduction $M \rar M'$ is 
$\beta_{\mathit{need}}$-rooted with argument $N$
if $M = \Theta[\app{(\fun{x}{N'})}{N}]$ and 
$M' = \Theta[\letrec{\be{x}{N}}{N'}]$. A reduction sequence 
$M \rarr M'$ preserves a $\beta_{\mathit{need}}$-root with
argument $N$ if none of (one-step) reductions in the
sequence is $\beta_{\mathit{need}}$-rooted with argument $N$.
The following lemma is proved similarly to Lemma~\ref{lemma:red_app}.

\begin{lemma}\label{lemma:red_app_rec}
For any $\Theta$, $M$ and $N$ such that
$\Theta[\app{M}{N}]$ is a program,
if $\Theta[\app{M}{N}] \rarr^n \Theta'[\app{V}{N}]$
and the reduction sequence preserves a $\beta_{\mathit{need}}$-root
with argument $N$, 
then $\Theta[M] \rarr^{n'} \Theta'[V]$ with $n' \leq n$.
\end{lemma}

Now we are ready to prove the reduction semantics correct with respect
to the natural semantics. 

\begin{proposition}\label{prop:red_then_eval_rec}
For any program $M$, if $M \rarr A$, 
then there exist $\Theta$ and $V$ such that $\Theta[V]$ and $A$ belong
to the same $\alpha$-equivalence class and \eval{}{M}{\decomp{\Theta}}{V}.
\end{proposition}
\begin{proof}
Without loss of generality, we assume 
$\Theta[V]$ and $A$ are syntactically identical. 
We prove by induction on the length of the reductions of $M$. 
Let $M = \Theta'[M']$ with $M' \not= \letrec{D}{N}$. We perform case analysis 
on $M'$.\\
- The case of an answer is obvious.\\
- Suppose $M =\app{M_1}{M_2}$ and we have:\\
\hspace*{3ex}$\begin{array}{l}
\Theta'[\app{M_1}{M_2}]
\rarr \Theta_1[\app{(\fun{x}{N})}{M_2}]
\rar \Theta_1[\letrec{\be{x}{M_2}}{N}]
\rarr \Theta[V]
\end{array}$\vspace{.5ex}\\
By Lemma~\ref{lemma:red_app_rec} and ind. hyp., \eval{}{\Theta'[M_1]}
{\decomp{\Theta_1}}{\fun{x}{N}}.
By ind. hyp.,  \\
\eval{}{\Theta_1[\letrec{\be{x}{M_2}}{N}]}{\decomp{\Theta}}{V}. 
Thus we deduce \eval{}{\Theta'[\app{M_1}{M_2}]}{\decomp{\Theta}}{V}.\\
- The case where $M =\app{M_1}{M_2}$ and $M_1$ reduces to \error\ is similar to the above
case.\\
- Suppose $M = x$ and $\Theta = \Theta_1,\letrec{\be{x}{N},D}{[]},\Theta_2$ 
and we have:\\
\hspace*{3ex}$\begin{array}{l}
\Theta_1[\letrec{\be{x}{N},D}{\Theta_2[x]}]\\
\rarR^n\ \Theta'_1[\letrec{\be{x}{\Theta_3[V]},D_1}{\Theta_2[x]}]\\
\rarr \Theta'_1[\letrec{\be{x}{V}, \flatten{\Theta_3}, D_1}{\Theta_2[x]}]\\
\rar \Theta'_1[\letrec{\be{x}{V}, \flatten{\Theta_3}, D_1}{\Theta_2[V]}]
\end{array}$\\
By Lemma~\ref{lemma:red_letrec}, 
$\Theta_1[\letrec{\be{x}{\error},D}{N}]
\rarr^n\ \Theta'_1[\letrec{\be{x}{\error}, D_1}{\Theta_3[V]}]$.
By ind. hyp., \eval{}{\Theta_1[\letrec{\be{x}{\error},D}{N}]}
{\decomp{\Theta'_1,\letrec{\be{x}{\error}, D_1}{[]}, \Theta_3}}{V}.
By Lemma~\ref{lemma:eval_heap_mono_rec}, 
\eval{\decomp{\Theta_1,\letrec{\be{x}{\error},D}{[]},\Theta_2}}{N}
{\decomp{\Theta'_1,\letrec{\be{x}{\error}, D_1}{[]}, \Theta_3,\Theta_2}}{V}.
Thus we deduce
\eval{}{\Theta_1[\letrec{\be{x}{N},D}{\Theta_2[x]}]}
{\decomp{\Theta'_1,\letrec{\be{x}{V}, \flatten{\Theta_3}, D_1}{[]}, \Theta_2}}{V}.
\end{proof}

Collecting all propositions together, 
we prove equivalence of the two semantics. 

\begin{theorem}\label{theorem:equiv_cyclic}
For any program $M$, the following two conditions hold:
\begin{enumerate}
\item if $M \rarr A$ then there exist $\Theta$ and $V$ such that
$\Theta[V]$ and $A$ belong to the same $\alpha$-equivalence class and 
\eval{}{M}{\decomp{\Theta}}{V}
\item if \eval{}{M}{\Psi}{V} then $M \rarr \Theta[V]$ where $\decomp{\Theta} = \Psi$.
\end{enumerate}
\end{theorem}
\begin{proof}
1: By Proposition~\ref{prop:red_then_eval_rec}. 
2: By Proposition~\ref{prop:eval_then_Eval_rec}
and Lemma~\ref{lemma:Eval_heap_grow_rec}, 
\Eval{}{M}{\Theta}{V} with $\decomp{\Theta} = \Psi$.
By Proposition~\ref{prop:Eval_then_red_rec}, $M \rarr \Theta[V]$.
\end{proof}

\subsection{Adequacy}\label{sec:byname}

In this subsection we state that the natural semantics is adequate
using a denotational semantics in the style of Launchbury
\cite{bigstepforlazy}. 
We adapt his proof strategy with minor modifications. 
A gentle explanation of the strategy is referred to his paper. 

\medskip

We define the denotational semantics for {\it pure} expressions of
\cyclic. A program $M$ is {\it pure} if it does not contain black
holes. The denotational semantics models functions
by a lifted function space \cite{fullabsforlazy}. 
We represent lifting using \lift, and
projection using \proj\ (written as a postfix operator). 
Let \Values\ be some appropriate domain containing at least a lifted
version of its own function space. {\it Environments}, ranged over by
$\rho$, are functions from \Vars\ to \Values, where 
\Vars\ denotes the infinitely many set of variables of \cyclic. 
The notation \supp{\rho} denotes the support of $\rho$, or
$\supp{\rho} = \{x \mid \rho(x) \not= \bot\}$. 
The notation 
$\{\maps{x_1}{z_1}, \ldots, \maps{x_n}{z_n}\}$ where $z_i$'s are
elements of \Values\ denotes an environment $\rho$ such that
\supp{\rho} = $\{x_1, \ldots, x_n\}$ and $\rho(x_i)$ = $z_i$ 
for all $i$ in $1, \ldots, n$.
The notation $\rho_{\bot}$ denotes an ``initial'' environment which maps all
variables to \undef, i.e. \supp{\rho_{\bot}} = $\emptyset$.

The semantic functions \semax{M}{\rho} and \Semax{D}{\rho}
respectively give meanings to the expression $M$ and the bindings 
$D$ under the  environment $\rho$. The former returns an element from
{\it Value} and the latter an environment. They are defined by mutual
recursion as follows:
\[
\begin{array}{l}
\semax{\fun{x}{M}}{\rho} = \lift\ (\fun{\nu}{\semax{M}{\rho\upp\{\maps{x}{\nu}\}}})\\
\semax{\app{M}{N}}{\rho} = (\semax{M}{\rho}) \proj (\semax{N}{\rho})\\
\semax{x}{\rho} = \rho(x)\\
\semax{\letrec{\be{x_1}{M_1},\ldots, \be{x_n}{M_n}}{N}}{\rho}
= \semax{N}{\Semax{\be{x_1}{M_1},\ldots, \be{x_n}{M_n}}{\rho}}\\
\Semax{\be{x_1}{M_1},\ldots, \be{x_n}{M_n}}{\rho} 
= \fix{\rho'}{\rho\upp\{\maps{x_1}{\semax{M_1}{\rho'}}, \ldots, \maps{x_n}{\semax{M_n}{\rho'}}\}}
\end{array}
\]
where $\mu$ denotes the least fixed point operator.  
\Semax{D}{\rho} is defined only when $\rho$ is consistent with $D$,
i.e., if $\rho$ and $D$ bind the same variable, then they maps the
variable to values for which an upper bound exists. 
The semantic function for heaps is defined in the same way as that for
bindings by identifying a heap with an unordered sequence of
bindings. 

We define an order on environments such that \envleq{\rho}{\rho'} if
for all $x$ in \supp{\rho}, $\rho(x) = \rho'(x)$.

We revise the natural semantics for \cyclic\ so that it gets
stuck when direct cycles are encountered as in Launchbury's semantics. 
Therefore we replace the {\it Variable} rule of
figure \ref{fig:eval_rec} by the following alternative:
\[
\infer{
  \evaltotal{\Psi}{x}{\update{\Phi}{x}{V}}{V}{X}
}{
  x \in \dom{\Psi}
  &\evaltotal{\restr{\Psi}{x}}{\Psi(x)}{\Phi}{V}{X\cup \{x\}}
}
\]
The notation \restr{\Psi}{x} denotes the restriction of $\Psi$
to $\dom{\Psi}\backslash \{x\}$. We use $\downarrow$ instead of
$\Downarrow$ to denote the revised semantics.

\begin{lemma}\label{lemma:eval_iff_evaltotal}
For any pure expression $M$,
\eval{}{M}{\Psi}{\fun{x}{N}} iff \evaltotal{}{M}{\Psi}{\fun{x}{N}}{\emp}.
\end{lemma}

A heap $\Psi$ is pure if for all $x$ in \dom{\Psi}, $\Psi(x)$ is
pure. A configuration \evalin{\Psi}{M} is pure if both $\Psi$ and $M$ are
pure.

\begin{lemma}\label{lemma:evaltotal_correct_to_byname}
If \evalin{\Psi}{M} is pure and 
\evaltotal{\Psi}{M}{\Phi}{V}{X},
then for any environment $\rho$,
\semax{M}{\Semax{\Psi}{\rho}} = \semax{V}{\Semax{\Phi}{\rho}}
and  \envleq{\Semax{\Psi}{\rho}}{\Semax{\Phi}{\rho}}.
\end{lemma}

The following proposition 
states that derivations preserve non-bottom meanings of pure expressions.

\begin{proposition}\label{prop:eval_correct_to_byname}
For any pure program $M$, if  \eval{}{M}{\Psi}{\fun{x}{N}}
then \semax{M}{\rho_{\bot}} = \semax{\fun{x}{N}}{\Semax{\Psi}{\rho_{\bot}}}.
\end{proposition}
\begin{proof}
By Lemma \ref{lemma:eval_iff_evaltotal}, 
\evaltotal{}{M}{\Psi}{\fun{x}{N}}{\emp}.
By Lemma \ref{lemma:evaltotal_correct_to_byname},
\semax{M}{\rho_{\bot}} = \semax{V}{\Semax{\Psi}{\rho_{\bot}}}
\end{proof}

Next we characterize when derivations exist. 

\begin{lemma}
If \evalin{\Psi}{M} is pure and
\evaltotal{\Psi}{M}{\Phi}{\fun{x}{N}}{X} 
then \semax{M}{\Semax{\Phi}{\rho_{\bot}}} $\not= \bot$.
\end{lemma}

\newcommand{\apsemax}[2]{\ensuremath{\mathcal{N}[\![#1]\!]_{#2}}}

Following Launchbury, we define a resourced denotational semantics. 
Let $C$ be the countable chain domain defined as the least solution to
the domain equation $C = C_{\bot}$. We represent lifting in $C$ by
injection function $S: C \rar C$ and limit element $S(S(S \ldots))$ by
$\omega$. {\it Resourced environments}, ranged over by $\sigma$, 
are functions from \Vars\ to functions from $C$ to \Values, i.e.,
$\sigma: \Vars \rar\ (C \rar\ \Values)$.
We define a resourced semantic function \apsemax{M}{\sigma} as follows:
\[
\begin{array}{l}
\apsemax{M}{\sigma} ~\bot = \bot\\
\apsemax{\fun{x}{M}}{\sigma} ~(S ~k) =
\lift\ (\fun{\nu}{\apsemax{M}{\sigma \upp \{x \mapsto \nu\}}})\\
\apsemax{\app{M}{N}}{\sigma}~ (S ~k) =
(\apsemax{M}{\sigma} ~k) \proj (\apsemax{N}{\sigma}) ~k\\
\apsemax{x}{\sigma} ~(S ~k) = \sigma ~x ~k\\
\apsemax{\letrec{\be{x_1}{M_1},\ldots,\be{x_n}{M_n}}{M}}{\sigma} ~(S~k)
= \\
\apsemax{M}
{\fix{\sigma'}{\sigma \upp \{x_1 \mapsto \apsemax{M_1}{\sigma'},
\ldots, x_n \mapsto \apsemax{M_n}{\sigma'}\}}} ~k
\end{array}
\]

We define an alternative natural semantics in which {\it Variable}
rule is replaced  
by
\[
\infer{
  \evalname{\Psi,\maps{x}{M}}{x}{\Phi}{V}
}{
  \evalname{\Psi,\maps{x}{M}}{M}{\Phi}{V}
}
\]
We use $\downarrow_{\mathit{name}}$ to denote
this alternative semantics. 

\begin{lemma}\label{lemma:evalname_then_eval}
For any pure expression $M$, if \evalname{}{M}{\Psi}{\fun{x}{N}} then
\eval{}{M}{\Psi'}{\fun{x}{N}}.
\end{lemma} 

\begin{lemma}\label{lemma:apsemax_correct_to_evalname}
For any pure expressions $M, M_1, \ldots, M_n$, \\
if \apsemax{M}{\fix{\sigma}
{\{\maps{x_1}{\apsemax{M_1}{\sigma}},\ldots,\maps{x_n}{\apsemax{M_n}{\sigma}}\}}}
~$(S^m ~\bot) \not= \bot$, then
\evalname{\maps{x_1}{M_1},\ldots,\maps{x_n}{M_n}}{M}{\Psi}{\fun{x}{N}}.
\end{lemma}

The following proposition states that a pure
expression evaluates to an abstraction if and only if its meaning is
a non-bottom element. Since the natural semantics is
deterministic, we can 
deduce that if a pure expression 
evaluates to a black hole then its meaning is a bottom element.

\begin{proposition}\label{prop:eval_adequate_to_byname}
For any pure program $M$, 
$\semax{M}{\rho_{\bot}} \not= \bot$ iff 
\eval{}{M}{\Psi}{\fun{x}{N}}.
\end{proposition}
\begin{proof}
If: There exists $m$ such that
\apsemax{M}{\sigma_{\bot}} $(S^m ~\bot)$ \noteq\ \undef.
By Lemma \ref{lemma:apsemax_correct_to_evalname},
\evalname{}{M}{\Psi}{\fun{x}{N}}.
By Lemma \ref{lemma:evalname_then_eval},
\eval{}{M}{\Phi}{\fun{x}{N}}. Only if: By Proposition~\ref{prop:eval_correct_to_byname}.
\end{proof}

\section{An extension with pairs}\label{sec:pair}

\begin{figure}[t]
\hspace*{\fill}{\small $\begin{array}{llcl}
{\it Expressions}\hspace*{5ex}& 
M,N &::=& \prd{M}{N} \mid \prj{M}{i} \mid \ldots \\
{\it Values}& V &::=& \prd{V_1}{V_2} \mid \ldots \\
{\it Contexts}&
E &::=& \prd{E}{M} \mid \prd{V}{E} \mid \prj{E}{i} \mid \ldots\\
\end{array}$}\hspace*{\fill}\vspace{1ex}
\caption{Extension with pairs}\label{fig:syntax_pair}
\end{figure}

In this section we extend the cyclic calculus \cyclic\ with (eager)
pairs.  The motivation for the extension is to set up a basic
framework to study lazy recursive records. Lazy evaluation is used in
some programming languages to evaluate recursive records. Hence we
think the extension is worth considering.

\begin{figure}[t]
\hspace*{\fill}{\small $\begin{array}{ll}
{\it prj}: 
&\prj{\prd{V_1}{V_2}}{i} ~\needmapsto\ V_i\\
{\it lift}_{\pi}:
&\prj{\letrec{D}{A}}{i}  ~\needmapsto\ \letrec{D}{\prj{A}{i}}\\
\mathit{lift}_{\mathit{pair}_1}:
& \prd{(\letrec{D}{A})}{M} ~\needmapsto\ \letrec{D}{\prd{A}{M}}\\
\mathit{lift}_{\mathit{pair}_2}:
& \prd{V}{\letrec{D}{A}} ~\needmapsto\ \letrec{D}{\prd{V}{A}}
\end{array}$}\hspace*{\fill}\vspace{1ex}
\caption{Reduction semantics for pairs}\label{fig:red_pair}
\end{figure}

\begin{figure}[t]
{\small 
\[
\begin{array}{c}
\mathit{Pair}\\
\infer{
  \eval{\Psi}{\prd{M_1}{M_2}}{\Psi_2}{\prd{V_1}{V_2}}
}{
  \eval{\Psi}{M_1}{\Psi_1}{V_1}
  &\eval{\Psi_1}{M_2}{\Psi_2}{V_2}
}\vspace{1ex}\\
\mathit{Projection}\\
\infer{
  \eval{\Psi}{\prj{M}{i}}{\Phi}{V_i}
}{
  \eval{\Psi}{M}{\Phi}{\prd{V_1}{V_2}}
}
\end{array}
\]}
\caption{Natural semantics for pairs}\label{fig:eval_pair}
\end{figure}

To accommodate pairs, we extend the syntax of \cyclic\ as given in
figure~\ref{fig:syntax_pair}. 
Now an expression may be a pair \prd{M}{N} or projection \prj{M}{i}.
A value may be a pair of values \prd{V_1}{V_2}.
Evaluation contexts contain three new productions 
\prd{E}{M}, \prd{V}{E}  and \prj{E}{i}.
Pairs are evaluated eagerly from left to right.

Figures~\ref{fig:red_pair} and~\ref{fig:eval_pair} respectively give
new rules to be added to the reduction and the evaluation semantics, for
evaluating and destructing pairs.  The two rules in
figure~\ref{fig:eval_pair} and {\it prj} in figure~\ref{fig:red_pair}
should be self-explanatory.  Heap reconfiguration is implicit in the
evaluation semantics, but is explicit in the reduction semantics. That
is, $\mathit{lift}_{\pi}$ is hidden in {\it Projection}, and
$\mathit{lift}_{\mathit{pair}_1}$ and
$\mathit{lift}_{\mathit{pair}_2}$ are in {\it Pair}.  The equivalence
result of the two semantics straightforwardly carries over to the
extension.

\begin{theorem}
For any program $M$, the following two conditions hold:
\begin{enumerate}
\item if $M \rarr A$ then there exist $\Theta$ and $V$ such that
$\Theta[V]$ and $A$ belong to the same $\alpha$-equivalence class and 
\eval{}{M}{\decomp{\Theta}}{V}
\item if \eval{}{M}{\Psi}{V} then $M \rarr \Theta[V]$ where $\decomp{\Theta} = \Psi$.
\end{enumerate}
\end{theorem}

\newcommand{\valrar}{%
\mbox{$\,\displaystyle\mathop{\rightarrow}_{\mbox{\tiny\sf value}}\,$}}
\newcommand{\valrarr}{%
\mbox{$\,\displaystyle\mathop{\twoheadrightarrow}_{\mbox{\tiny\sf value}}\,$}}
\newcommand{\needrarr}{%
\mbox{$\,\displaystyle\mathop{\twoheadrightarrow}_{\mbox{\tiny\sf need}}\,$}}

\section{Call-by-value letrec calculus \cyclicval}\label{sec:byvalue}

The {\sf delay} and {\sf force} operators as provided in Scheme
\cite{scheme}, or OCaml's equivalent {\sf lazy} and {\sf force}
\cite{OCaml}, can be emulated by \letrec{\be{x}{M}}{\fun{x'}{x}} for
$\mathsf{delay}(M)$ and \app{M}{(\fun{x}{x})} for $\mathsf{force}(M)$.
It is crucial for this encoding that letrec-bindings are evaluated
lazily.  However, in the presence of ML's traditional value recursion
restriction, which requires the right-hand side of recursive bindings
to be syntactic values, lazy {\sf letrec}'s are faithful to ML's {\sf
letrec}'s.  Note that $\mathsf{delay}(M)$ is considered to be a
syntactic value.  Therefore we are interested in a call-by-value
variant of \cyclic, which can model a call-by-value letrec lambda
calculus with {\sf delay/force} operators. 
For instance Syme's initialization
graphs~\cite{initgraphs}, which 
underlie the object initialization strategy of F\#~\cite{FSharp},
fit in this variant extended with $n$-tuples, or records. 

\begin{figure}
\hspace*{\fill}{\small $\begin{array}{llcl}
{\it Expressions}\hspace*{5ex}& 
M,N &::=& x \mid \fun{x}{M} \mid \app{M}{N} \mid \letrec{D}{M} \mid \error\\
{\it Definitions}&
D &::=& \rt \mid D, \be{x}{M}\\
{\it Values}& V &::=& \fun{x}{M} \mid \error\\
{\it Answers}& 
A &::=& V \mid \letrec{D}{A}\\
{\it Good~Answers}& 
G &::=& \fun{x}{M} \mid \letrec{D}{G}\\
\mbox{{\it By-value~Contexts}}&
E &::=& [] \mid \app{E}{M} \mid \app{V}{E} \mid \letrec{D}{E}\\ 
& &\mid& \letrec{x = E, D}{E'[x]}\\
& &\mid& \letrec{x' = E, \dps{x}{x'},D}{E'[x]}\\
{\it Dependencies}&
\dps{x}{x'} &::=& \be{x}{E[x']}\\ 
           &&\mid& \dps{x}{x''},\be{x''}{E[x']}\\
\end{array}$}\hspace*{\fill}\vspace{1ex}
\caption{Syntax of \cyclicval}\label{fig:syntax_val}
\end{figure}

\begin{figure}
\hspace*{\fill}{\small $\begin{array}{ll}
\beta_{\mathit{value}}:
&\app{(\fun{x}{M})}{(\fun{x'}{M'})} ~\valuemapsto\ \letrec{\be{x}{\fun{x'}{M'}}}{M}\\
\mathit{lift}_{\mathit{arg}}:
&\app{V}{(\letrec{D}{A})} ~\valuemapsto\ \letrec{D}{\app{V}{A}}\\
\mathit{error}_{\mathit{arg}}:
\hspace*{3ex}&\app{(\fun{x}{M})}{\error} ~\valuemapsto\ \error\\
\end{array}$}\hspace*{\fill}\vspace{1ex}
\caption{Reduction semantics for \cyclicval}\label{fig:red_val}
\end{figure}

In figure~\ref{fig:syntax_val} we define the syntax of \cyclicval, 
a call-by-value variant of \cyclic\footnote{It should be noted that 
the true beta-value axiom is  \app{(\fun{x}{M})}{V} = \subst{M}{V}{x}, 
as introduced by Plotkin.}.
It differs from \cyclic\ in that evaluation contexts contain the production
\app{V}{E} to force evaluation of arguments. We have introduced
{\it good answers} to distinguish successful termination, which returns abstraction;
we will use good answers to state Proposition~\ref{prop:cbvcbn}.
As for the reduction semantics, we replace $\beta_{\mathit{need}}$ with
$\beta_{\mathit{value}}$ and add two new rules $\mathit{lift}_ {\mathit{arg}}$
and $\mathit{error}_ {\mathit{arg}}$ 
as given in figure~\ref{fig:red_val}. Otherwise the reduction rules
are unchanged from figure~\ref{fig:red_rec}.
An expression $M$ {\it by-value reduces} to $N$, written $M \valrar N$, if
$M = E[M']$ and $N = E[N']$ where $M' \valuemapsto N'$.
We write \valrarr\ to denote the reflexive and transitive closure of \valrar.
To avoid confusion we write \needrarr, instead of \rarr, 
to denote multi-step reductions in \cyclic. 

\medskip

Proposition~\ref{prop:cbvcbn} states that \cyclic\ is more likely 
to return good answers than \cyclicval. This is not surprising. 
We prove the proposition by defining the natural semantics for \cyclicval\
and by relating \cyclicval\ and \cyclic\ in terms of the natural semantics.

\begin{proposition}\label{prop:cbvcbn}
For any program $M$, if 
$M \valrarr G$ then $M \needrarr G'$.
\end{proposition}

An expression which returns a black hole in \cyclicval\ may return 
abstraction in \cyclic, 
e.g. \letrec{\be{x}{\app{(\fun{y}{\fun{y'}{y}})}{x}}}{x}.

%

\section{Related work}\label{sec:related}
Our work builds on previous work by Launchbury \cite{bigstepforlazy},  
Sestoft \cite{lazymachine}, Ariola and Felleisen \cite{thecbnlambda}
and Maraist et al. \cite{cbnlambda}. The reduction semantics present
in the paper are mostly identical to those of Ariola and
Felleisen. As to the natural semantics for \acyclic, we revised 
that of Maraist et al. by correctly enforcing variable hygiene in
the style of Sestoft and by explicitly introducing an inference rule
for the let construct. As to the natural semantics for \cyclic, we
revised that of 
Sestoft by eliminating the precompilation step. Adequacy of the
natural semantics for \acyclic\ is ascribed to its correspondence with
the reduction semantics, which is proved equivalent to call-by-name by
Ariola and Felleisen. In turn we showed adequacy of the
natural semantics for \cyclic\ by adapting Launchbury's denotational
argument. Adequacy of the reduction semantics for
\cyclic\ is then ascribed to its correspondence with the natural semantics;
to the best of our knowledge, this fact has not been shown so far. 
In the above discussed sense, our work extends those previous
work. 



There are several lines of work which considers other styles of
formalization of call-by-need 
in the presence or absence of recursion. 
Below we review some of them. The reader may be interested in the
concluding remarks of \cite{cbnlambda}, where Maraist et al. discuss
the reduction semantics  in relation to other systems. 

Recent work by Garcia et al.~\cite{lazydcntrl} proposed an abstract machine 
for the let-free formulation of the acyclic calculus \acyclic, which
is proved equivalent to the reduction semantics of Ariola and
Felleisen \cite{thecbnlambda}.
They also presented a simulation of the machine by 
a call-by-value lambda calculus extended with delimited control operators. 
While developed independently, their abstract machine, in particular
the refined one, and our instrumented natural semantics bear similarities in that 
both manipulate sequenced evaluation contexts while retaining the
structural knowledge of a term that has been discovered. 
More thorough comparison might suggest a means of 
simulating the cyclic calculus \cyclic\ using delimited control. 
This is one direction for future work. 

Sestoft revised the natural semantics of Launchbury by enforcing
variable hygiene correctly and changing the $\alpha$-renaming strategy
\cite{lazymachine}.
He derived an abstract machine for call-by-need from the revised
semantics. 
The machine has a small-step semantics 
and uses global heaps to implement sharing of evaluation. 
Starting from a simple machine, he refines it to a more efficient machine
in several steps. 
The machine is proved equivalent to his natural semantics. 
As discussed earlier, the natural semantics for
\cyclic\ is strongly inspired by his semantics. 


Okasaki et al.~\cite{byneedcps} proposed a transformation of call-by-need
$\lambda$ terms, in the absence of recursion, into continuation-passing style, 
which is proved equivalent to a call-by-need continuation semantics. 
Sharing of evaluation is implemented by ML-style references, which resemble
global heaps.

Ariola and Klop~\cite{cycliclambda} and Ariola and
Blom~\cite{cycliccalculi} studied 
equational theories of cyclic lambda calculi
by means of cyclic lambda graphs. The former observed that having 
non-restricted substitution leads to non-confluence and proposed a restriction
on substitution to recover confluence. The latter proposed a relaxed notion 
of confluence which holds in the presence of non-restricted substitution.
In~\cite{cycliccalculi} a calculus supporting sharing is considered,
but a reduction strategy for the calculus is not studied. 

Danvy \cite{defunctionalized} advocates the use of abstract machines
as a "natural meeting ground" of various functional implementations of
operational semantics, especially the small-step reduction semantics
and big-step natural semantics. In a large perspective, our work
presented here can be thought as making an analogous case for a
destructive, non-functional setting, in which circularly shared
computation contributes significant complexities.


\section{Conclusion}\label{sec:conclusion}
We have presented natural semantics for acyclic and cyclic
call-by-need lambda calculi, which are proved equivalent to the
reduction semantics given by Ariola and Felleisen. 
We observed differences of the two styles of formalization
in the treatment of when to reorganize the heap structure and how to
focus redexes. 
The proof uses instrumented natural semantics as mediatory
semantics of the two, 
in order to bridge these differences by making heap reorganization and
redex focusing explicit.

This work is initially motivated to study lazy evaluation strategies
for recursive records in terms of the reduction semantics as well as
the natural semantics.  Therefore we have considered an extension with
eager pairs and a call-by-value variant with lazy letrec.

\subsubsection*{Acknowledgment}
We thank the anonymous referees for their careful reviewing and 
Matthias Felleisen for his editorial support. 

\bibliographystyle{plain}
\bibliography{references}

\begin{thebibliography}{10}

\bibitem{fullabsforlazy}
S.~Abramsky and C.-H.~L. Ong.
\newblock Full abstraction in the lazy lambda calculus.
\newblock {\em Information and Computation}, 105(2):159--267, 1993.

\bibitem{thecbnlambda}
Z.~Ariola and M.~Felleisen.
\newblock The {C}all-by-{N}eed {L}ambda {C}alculus.
\newblock {\em Journal of Functional Programming}, 7(3), 1997.

\bibitem{cycliccalculi}
Z.~M. Ariola and S.~Blom.
\newblock Cyclic {L}ambda {C}alculi.
\newblock In {\em Proc. {T}heoretical {A}spects of {C}omputer {S}oftware},
  volume 1281 of {\em Lecture Notes in Computer Science}, pages 77--106.
  Springer, 1997.

\bibitem{cycliclambda}
Z.~M. Ariola and J.~W. Klop.
\newblock Cyclic lambda graph rewriting.
\newblock In {\em Proc. {S}ymposium on {L}ogic in {C}omputer {S}cience}, pages
  416--425, 1994.

\bibitem{defunctionalized}
O.~Danvy.
\newblock Defunctionalized {I}nterpreters for {P}rogramming {L}anguages.
\newblock In {\em Proc. {I}nternational {C}onference on {F}unctional
  {P}rogramming}. ACM Press, 2008.

\bibitem{lazydcntrl}
R.~Garcia, A~Lumsdaine, and A.~Sabry.
\newblock Lazy {E}valuation and {D}elimited {C}ontrol.
\newblock In {\em Proc. the {ACM SIGPLAN-SIGACT} {S}ymposium on the
  {P}rinciples of {P}rogramming {L}anguages}. ACM Press, 2009.

\bibitem{bigstepforlazy}
J.~Launchbury.
\newblock A {N}atural {S}emantics for {L}azy {E}valuation.
\newblock In {\em Proc. the {ACM SIGPLAN-SIGACT} {S}ymposium on the
  {P}rinciples of {P}rogramming {L}anguages}, 1993.

\bibitem{OCaml}
X.~Leroy, D.~Doligez, J.~Garrigue, D.~R\'{e}my, and J.~Vouillon.
\newblock The {O}bjective {C}aml system, release 3.11.
\newblock Software and documentation available on the Web,
  \url{http://caml.inria.fr/}, 2008.

\bibitem{cbnlambda}
J.~Maraist, M.~Odersky, and P.~Wadler.
\newblock A {C}all-by-{N}eed {L}ambda {C}alculus.
\newblock {\em Journal of {F}unctional {P}rogramming}, 8(3), 1998.

\bibitem{byneedcps}
C.~Okasaki, P.~Lee, and D.~Tarditi.
\newblock Call-by-need and {C}ontinuation-passing {S}tyle.
\newblock {\em {LISP} and {S}ymbolic {C}omputation}, 7, 1994.

\bibitem{cbn_cbv_lambda}
G.~Plotkin.
\newblock {Call-by-Name, Call-by-Value and the $\lambda$-Calculus}.
\newblock {\em {Theoretical Computer Science}}, 1(2):125--159, 1975.

\bibitem{lazymachine}
P.~Sestoft.
\newblock Deriving a lazy abstract machine.
\newblock {\em Journal of {F}unctional {P}rogramming}, 7(3):231--264, 1997.

\bibitem{scheme}
M.~Sperber, R.~K. Dybvig, M.~Flatt, and A.~V. Straaten.
\newblock {Revised$^6$ Report on the Algorithmic Language Scheme}.
\newblock Available at \url{http://www.r6rs.org/}, 2007.

\bibitem{initgraphs}
D.~Syme.
\newblock Initializing {M}utually {R}eferential {A}bstract {O}bjects: {T}he
  {V}alue {R}ecursion {C}hallenge.
\newblock In {\em Proc. {W}orkshop on {ML}}, 2005.

\bibitem{FSharp}
D.~Syme and J.~Margetson.
\newblock {T}he {F}\# {P}rogramming {L}anguage, 2008.
\newblock Software and documentation available on the Web,
  \url{http://research.microsoft.com/en-us/um/people/curtisvv/fsharp_default.a%
spx}.

\bibitem{wf-ic-94}
A.~K. Wright and M.~Felleisen.
\newblock {A Syntactic Approach to Type Soundness}.
\newblock {\em {Information and Computation}}, 115(1):38--94, 1994.

\end{thebibliography}

\end{document}